%% file: StructureMain.tex
\documentclass[11pt]{article}
\usepackage[shortlabels]{enumitem}
\usepackage{array}
\usepackage{tabularx}
\usepackage{wrapfig}
\usepackage{float}
\usepackage{graphicx}
\usepackage[margin=1in]{geometry} 
\usepackage{comment}
\usepackage{amsmath,amsthm,amssymb}
\usepackage{hyperref}
\usepackage{algorithm}
\usepackage[noend]{algpseudocode}
\usepackage{tikz}
\input{insbox}

\usepackage{enumitem}
\SetEnumitemKey{midsep}{topsep=3pt, itemsep=0pt, partopsep=0pt}
\setlist{midsep}

\usepackage[noframe]{showframe}
\usepackage{framed}
\usepackage{lipsum}
\usetikzlibrary{automata, positioning, calc, fit, shapes.misc}

\usepackage{etoolbox}
\makeatletter
\AfterEndEnvironment{algorithm}{\let\@algcomment\relax}
\AtEndEnvironment{algorithm}{\kern2pt\hrule\relax\vskip3pt\@algcomment}
\let\@algcomment\relax
\newcommand\algcomment[1]{\def\@algcomment{#1}}

\renewcommand\fs@ruled{\def\@fs@cfont{\bfseries}\let\@fs@capt\floatc@ruled
  \def\@fs@pre{\hrule height.8pt depth0pt \kern2pt}%
  \def\@fs@post{}%
  \def\@fs@mid{\kern2pt\hrule\kern2pt}%
  \let\@fs@iftopcapt\iftrue}
\makeatother

 {\endMakeFramed}
\definecolor{shadecolor}{gray}{0.90}

\makeatletter
\newcommand\Ps@textstyle[2]{\mathbb{P}_{#1}\left[{#2}\right]}
\newcommand\Es@textstyle[2]{\mathbb{E}_{#1}\left[{#2}\right]}

\newcommand\Ps[2]{%
  \mathchoice % special stiling in display mode, regular elsewhere.
  {\underset{{#1}}{\mathbb{P}}\left[{#2}\right]
  }{\Ps@textstyle{#1}{#2}}{\Ps@textstyle{#1}{#2}}{\Ps@textstyle{#1}{#2}}
}
\newcommand\Es[2]{%
  \mathchoice % special stiling in display mode, regular elsewhere.
  {\underset{{#1}}{\mathbb{E}}\left[{#2}\right]
  }{\Es@textstyle{#1}{#2}}{\Es@textstyle{#1}{#2}}{\Es@textstyle{#1}{#2}}
}
\makeatother

\newcommand{\M}{\mathcal{M}}
\newcommand{\Mmatched}{\mathcal{M}_{\mathrm{matched}}}
\newcommand{\Msingle}{\mathcal{M}_{\mathrm{single}}}
\newcommand{\W}{\mathcal{W}}
\newcommand{\Wmatched}{\mathcal{W}_{\mathrm{matched}}}
\newcommand{\Wsingle}{\mathcal{W}_{\mathrm{single}}}
\renewcommand{\L}{\mathcal{L}}

\newtheorem{definition}{Definition}[section]

\newtheorem{corollary}[definition]{Corollary}
\newtheorem{theorem}[definition]{Theorem}
\newtheorem{claim}[definition]{Claim}

\tikzset{cross/.style={cross out, draw=black, minimum size=2*(#1-\pgflinewidth), inner sep=0pt, outer sep=0pt},
cross/.default={1pt}}
\begin{document}

\title{
  Representing All Stable Matchings by Walking a Maximal Chain
}
\author{
  Linda Cai\\
  tcai@princeton.edu
  \and
  Clay Thomas\\
  claytont@princeton.edu
}
\maketitle

\begin{abstract}
  % Stable matchings are a classical area of study, and
  % the well-known book of Gusfield and
  % Irving~\cite{GusfieldStableStructureAlgs89} gives strong structural and
  % algorithmic results useful for understanding this problem. 
  The seminal book of Gusfield and
  Irving~\cite{GusfieldStableStructureAlgs89} provides a compact and
  algorithmically useful way to represent the collection of stable matches
  corresponding to a given set of preferences.
  In this paper, we reinterpret the main results
  of~\cite{GusfieldStableStructureAlgs89}, giving
  a new proof of the characterization which is able to bypass a lot
  of the ``theory building'' of the original works.
  We also provide a streamlined and efficient way
  to compute this representation.
  Our proofs and algorithms emphasize the connection to
  well-known properties of the deferred acceptance algorithm.
\end{abstract}

\section{Introduction} \label{sectionIntro}

\input{Introduction.tex}

  % \input{IntroductionRewrite.tex}

%%%%%%%%% %%%%%%%%% %%%%%%%%% %%%%%%%%% %%%%%%%%% %%%%%%%%%
% \part{The Set of Stable Matchings}
%%%%%%%%% %%%%%%%%% %%%%%%%%% %%%%%%%%% %%%%%%%%% %%%%%%%%%

\section{Stable Matchings and Deferred Acceptance} \label{sectionReview}

\input{DifferedAcceptance.tex}

  \subsection{The Lattice of Stable Matchings}\label{subsectionLatticeIntro}

\input{LatticeStructure.tex}

%%%%%%%%% %%%%%%%%% %%%%%%%%% %%%%%%%%% %%%%%%%%% %%%%%%%%%
% \part{The Structure of Stable Matchings}
%%%%%%%%% %%%%%%%%% %%%%%%%%% %%%%%%%%% %%%%%%%%% %%%%%%%%%

\section{Navigating the Lattice of Stable Matchings}\label{sectionLattice}

\input{LatticeExploration.tex}

\section{Rotations} \label{sectionRotation}

\input{AbstractRotations.tex}

\section{Efficiently Finding Rotations} \label{sectionAlgo}

  \input{AlgAndComments.tex}

  \input{AlgCorrectness.tex}

  % \input{Conclusion.tex} 

 % \input{old_section_6.tex}

    % It's not clear if the authors of \cite{AshlagiUnbalancedCompetition17} realize this.

  % \begin{proof}
  %   Induct on the length of $V$ and the size of $S$.
  %   If $|V| = 1$, then 
  %   UUGH well this still gets a bit messy. Probably best to use known proofs.
  %   I think it would work IF you show a converse of the ``lattice neighbors
  %   covering'' thing: all guys hasse-adjacent in the lattice differ by an
  %   improvement cycle.
  % \end{proof}

\bibliography{MatchingMechs}{}
\bibliographystyle{alpha}

\appendix

\section{Contrasting Our Proofs With Those of~\cite{GusfieldStableStructureAlgs89}}
  \label{sectionContrastGI}

\input{GusfieldContrast.tex}

\section{A Minor Error in \cite{GusfieldStableStructureAlgs89}}
  \label{sectionCorrection}
  \input{GusfieldLabelsCounterex.tex}

\end{document}

%% file: Introduction.tex
Stable matching mechanisms are ubiquitous in theory and in practice,
especially in the ``bipartite case'' where agents lie in two disjoint groups
and one-to-one matches are made between members of different groups.
% See, for example, 
% and~\ref{roth2002economist} for real-world applications.
% Stable matchings present one of the rare examples in mechanism design without
% money where many positive results are possible.
The most commonly used stable matching mechanism is the Gale-Shapley algorithm,
i.e. ``one-side proposing deferred acceptance''.
This algorithm has the nice properties of being simple to implement,
fast to execute, and strategyproof for the proposing side.
% \cite{DubinsMachiavelliGaleShapley81}.
However, deferred acceptance always returns the best stable
matching for the proposing side and the worst stable match for the receiving side.
% In the traditional example of matching men to women, this distinguishes
% two stable matchings: the man-optimal and the woman-optimal outcome.
This leads to a basic question: what lies in between?
% (moreover, it is the unique mechanism satisfying strategyproofness for 
% one of the sides \cite{GaleMsMachiavelli85}).

% \subsection{The Structure of the Set of Stable Matchings}

This question can be rephrased as follows:
how can one understand, represent, and traverse the set of all 
stable matchings for a given set of preference?
An excellent answer to this question was given
by~\cite{GusfieldStableStructureAlgs89}, based on the
works~\cite{irving1986complexity, irving1985efficient, irving1987efficient}.
Despite the fact that there can be exponentially many stable
matchings\footnote{
  For an easy example, consider $n/2$ ``disjoint copies'' of an instance with
  two men and two women which has two stable matchings.
  This has $2^{n/2}$ stable outcomes.
},
the collection of \emph{all} stable matching
can be compactly represented in a form which is efficient to construct and
algorithmically useful, and sheds light on the structure of the stable
matching instance.

  % While the treatment in~\cite{GusfieldStableStructureAlgs89} is quite elegant,
  % in our opinion it is somewhat unnecessarily abstract and formal. %hard to approach.
  % We draw inspiration for many of our proofs from~\cite{ImmorlicaHonestyStability05}.
  % A major technical tool in our proofs is using the differed acceptance
  % algorithm to understand the collection of all stable matchings.
  % , as we use basic properties of 
  % the differed acceptance algorithm to understand the \emph{differences}
  % between stable matchings.
  % , closely resembling the main algorithm of~\cite{AshlagiUnbalancedCompetition17}.

  In this paper, we reinterpret and simplify the classification provided
  by~\cite{GusfieldStableStructureAlgs89}.
  We provide full proofs which characterize the ``lattice structure'' 
  of the set of stable matchings and culminate in a theorem equivalent to
  the main characterization of~\cite{GusfieldStableStructureAlgs89}:

  \begin{theorem}[Combination of theorems~\ref{theoremRotationsBijection}
    and~\ref{theoremAlgRuntime}]

    For any stable matching instance, there is a 
    directed acyclic graph $G$,
    computable in $O(n^2)$ time,
    such that there is a bijection between the set of all stable matchings and
    the collection of closed
    subsets of $G$ (i.e. the subsets $S$ of vertices of $G$ such that
    no directed edge $(u,v)$ of $G$ has $u\in S$ but $v\notin S$).
  \end{theorem}

  There is a compelling interpretation of the vertices of
  $G$. They are called ``rotations'', and represent the fact that, starting from
  some stable matching, some set of men $(m_0, \ldots, m_{k-1})$ can ``cyclically
  move partners'' (i.e. each $m_i$ gets re-matched to the partner of $m_{i+1}$
  (with indices mod $k$)) to arrive at a new stable matching.
  A full description of the rotations, and the dependencies between them,
  is given in this paper.
  Indeed, the primary simplifying contribution of this paper is in focusing on
  rotations ``from the start'' instead of going through other notions.

  We believe that our proof strategies and presentation is more intuitive and more
  ``fundamentally algorithmic'', as we focus on how simple
  properties of the ubiquitous differed acceptance algorithm can lead us
  to understand the full set of stable matchings.
  Furthermore, we give a new perspective on the algorithm used
  in~\cite{GusfieldStableStructureAlgs89} to construct the compact
  representation.
  Along the way, we correct a minor error in the original algorithm
  from~\cite{GusfieldStableStructureAlgs89} (for details,
  see appendix~\ref{sectionCorrection}).

  While the stable matching problem is a classic and well-studied problem,
  there are many exciting contemporary developments in the theory,
  from worst-case upper bounds on the number of stable matchings
  \cite{KarlinExpUBNumberStable18} to
  the communication complexity of finding stable 
  matchings~\cite{gonczarowski2019stable} to
  % surprising average-case behavior in unbalanced markets
  % \cite{AshlagiUnbalancedCompetition17} to
  detailed studies of different incentives properties
  \cite{AshlagiStabOSP18,Gonczarowski14}.
  Our intent is for this paper to provide a starting point for researchers
  interested in studying stable matching markets from an algorithmic perspective.
  % We believe that this paper is an excellent starting point for someone
  % interested in studying matching markets from a theoretical perspective.

  % In appendix~\ref{sectionRandom}, we present new intuition and proofs
  % for a preliminary version of the results of~\cite{AshlagiUnbalancedCompetition17},
  % which show that in unbalanced matching markets, the short side has a
  % significant advantage in every stable matching.

\subsection{Organization and relation to prior work}
  % Many parts of this paper can be read largely independently of the others.
  For completeness, we prove every result about stable matchings which we will
  need in this paper.
  In section~\ref{sectionReview}, we make our formal definitions
  and review the basic properties of deferred acceptance and the set of stable
  matchings. Readers familiar with stable matchings can likely skip this section
  (possibly reviewing the lattice-theoretic vocabulary given in
  section~\ref{subsectionLatticeIntro}).
  The core technical material is presented in sections~\ref{sectionLattice}
  to~\ref{sectionAlgo}.
  % , with section~\ref{sectionRotation} providing the main
  % characterization.
  \begin{itemize}
    \item In section~\ref{sectionLattice}, we discuss how to traverse the stable
      matching lattice algorithmically.
      Intuitively, this involves women ``rejecting'' their current match, and
      continuing running differed acceptance to get to a better stable matching
      (for the women).
      Our core technical tool, inspired by \cite{ImmorlicaHonestyStability05}
      is to use the concept of \emph{differed acceptance with truncated preferences}.

    \item In section~\ref{sectionRotation}, we define a compact representation
      of the stable matching lattice in terms of ``minimal differences'' called
      rotations, and prove that the representation is correct.
      Our definitions and theorems are as in~\cite{GusfieldStableStructureAlgs89},
      but we are able to significantly simplify our treatment by focusing
      on rotations ``from the start'' and avoiding intermediate representations.
      In particular, claim~\ref{claimToplSortsAreChains} and its proof using
      claim~\ref{claimToplSortsInterchange} are the key new ideas,
      which provide a way to show that a graph represents a lattice using a
      proof approach which (to the best of our knowledge) is brand new.
      Appendix~\ref{sectionContrastGI} provides a detailed
      comparison of our methods and those of~\cite{GusfieldStableStructureAlgs89}.

    \item In section~\ref{sectionAlgo}, we show how to efficiently
      construct the compact
      representation defined in section~\ref{sectionRotation}.
      While our algorithm is essentially
      equivalent to that in~\cite{GusfieldStableStructureAlgs89}
      (figure 3.2 on page 110),
      we provide a more streamlined way to find the ``predecessor relations''
      between rotations, which are the edges of the graph $G$,
      and thus avoid a minor error in the way
      that~\cite{GusfieldStableStructureAlgs89} finds these predecessor relations.
      In appendix~\ref{sectionCorrection}, we point out and correct this minor error.
      Our presentation is similar to that of the ``MOSM to WOSM'' algorithm
      in~\cite{AshlagiUnbalancedCompetition17},
      which relates more clearly to our conceptual use of differed acceptance.
  \end{itemize}

%% file: DifferedAcceptance.tex
  We start with the basic definitions.
  A matching market is a collection $\M$ of ``men'' and $\W$ of ``women'', where
  each man $m\in \M$ has a ranking over women in $\W$, represented
  as list ordered from most preferred to least preferred, and vice versa.
  Lists may be partial, and agents included on the list of some $a \in \M\cup\W$
  are called the acceptable partners of $a$.
  We write $w_1 \succ_m w_2$ if $w_1$ is ranked higher than $w_2$ on $m$'s list
  (or if $w_1$ is acceptable but $w_2$ is not ranked at all).
  We also denote the fact that $w$ is not an acceptable partner of $m$ by
  $\emptyset \succ_m w$, and conversely if $w$ is an acceptable partner of $m$
  we write $w \succ_m \emptyset$. A \emph{matching} is a set of vertex disjoint edges in the bipartite graph $G(M, W)$, where $(m, w) \in E(G)$ if and only if $m$ is acceptable to $w$ and vice versa. We denote a matching by $\mu : \M\cup \W\to \M\cup \W\cup\{\emptyset\}$, where $\mu(i)$ is the matched partner of agent $i$. We write $\mu(i) = \emptyset$ if
  agent $i$ is unmatched. 

  For a set of preferences $P = \{\succ_w\}_{w\in\W} \cup \{\succ_m\}_{m\in\M}$
  and any matching $\mu$, a man/woman pair
  $(m,w)$ is called \emph{blocking} if we simultaneously have
  $m \succ_w \mu(w)$ and $w \succ_m \mu(w)$.
  A matching $\mu$ is \emph{stable} for a set of preferences $P$
  if no unmatched man/woman pair is blocking for $P$.
  % we do not simultaneously have
  % $m \succ_w \mu(w)$ and $w \succ_m \mu(w)$.
  A \emph{pair} $(m,w)$ is called stable for $P$ if $\mu(m)=w$ in 
  \emph{some} stable matching,
  and $m$ is called a stable partner of $w$ (and vice-versa).

  % We also call stable matchings \emph{outcomes}.

\subsection{$MPDA$ and the man-optimal stable matching}

  The most natural way to find stable matchings is with
  the celebrated \emph{deferred acceptance algorithm}.
  In this paper, we consider man proposing deferred acceptance ($MPDA$)
  as given in Algorithm~\ref{algoMPDA}.
  For completeness, here we provide simple proofs of the basic
  but crucially important properties of this algorithm.

  \begin{algorithm}
    \caption{$MPDA$: Men-proposing deferred acceptance}
    \label{algoMPDA}
  \begin{algorithmic}[0]
    \State Let $U = \M$ be the set of unmatched men
    \State Let $\mu$ be an all empty matching
    \While { $U\ne \emptyset$ and some $m\in U$ has not proposed to every woman
      on his list}
      \State Pick such a $m$ (in any order)
      \State $m$ ``proposes'' to their highest-ranked woman $w$ which 
        they have not yet proposed to
      \If {$m \succ_w \mu(w)$} 
        \State If $\mu(w)\ne \emptyset$, add $\mu(w)$ to $U$
        \State Set $\mu(w) = m$, remove $m$ from $U$ 
      % \EndIf
      % \If {$\mu(h) = d_0$ and $d\succ_h d_0$} 
      %   \State Set $\mu(h) = d$, remove $d$ from $U$, add $d_0$ to $U$
      % \Else \ \ $h$ remains matched to $d_0$ and $d$ remains in $U$
      \EndIf
    \EndWhile
  \end{algorithmic}
  \end{algorithm}

  Intuitively, this algorithm starts with the men doing whatever they prefer
  the most, then doing the minimal amount of work to make the matching stable.
  Indeed, men propose in their order of preference. If a woman $w$ ever
  rejected a man $m$ they prefer over their current match,
  then \emph{remained} with their current match,
  then $(m,w)$ would clearly create an instability in the final matching.

  \begin{claim}
    The output of $MPDA$ is a stable matching.
  \end{claim}
  \begin{proof}
    First, observe that the $MPDA$ algorithm terminates
    because every man will propose to every woman at most once.
    The claim follows from two simple invariants of the algorithm:
    \begin{itemize}
      \item Men propose in their order of preference.
      \item Women can only increase the rank of their tentative match over time
        (and once they are matched, they stay matched).
    \end{itemize}
    Formally, consider a pair $m\in \M$, $w\in \W$ which
    is unmatched in the output matching
    $\mu$. Suppose for contradiction $w\succ_m \mu(m)$ and $m\succ_w \mu(w)$.
    In the $MPDA$ algorithm, $m$ would propose to $w$ before $\mu(m)$.
    This means that $w$ received a proposal from a man she preferred over her
    eventual match $\mu(w)$, a contradiction.
  \end{proof}

  Note that this algorithm gives us a very interesting existence result: it was
  not at all clear that stable matching existed before we had this algorithm.

  We can now formalize our intuition that $MPDA$ does the least amount of work
  needed to result in a stable outcome
  (after men propose to their favorite women).
  We show that every rejection which happens in $MPDA$ \emph{must}
  happen in order for a stable matching to result.
  The proof uses the following technique:
  although it's not immediately easy to show an
  event can't happen, you can show it \emph{can't happen for the first time}.

  \begin{claim}\label{claimRejectionUnstable}
    If a man $m\in \M$ is ever rejected by a woman $w\in \W$ during some run
    of $MPDA$ (that is, $m$ proposes to $w$ and $w$ does not accept) then no stable
    matching can pair $m$ to $w$.
  \end{claim}
  \begin{proof}
    Let $\mu$ be any matching.
    Suppose that some pair, matched in $\mu$, is rejected during $MPDA$.
    Consider the first time during in the run of $MPDA$ where such a rejection
    occurs, i.e. a woman $w$ rejects $\mu(w)$ but no other woman $w'$ has
    rejected $\mu(w')$ so far.
    In particular, let $w$ reject $m=\mu(w)$ in favor of $m'\ne m$
    (either because $m'$ proposed to $w$,
    or because $m'$ was already matched to $w$ and $m$ proposed).
    We have $m'\succ_w m$, so if $m'$ is unmatched in $\mu$, then $\mu$ is
    unstable.
    Thus we have $\mu(m') = w' \ne w$,
    and because this is the first time any man has been rejected by a match from
    $\mu$, $m'$ has not yet proposed to $w'$.
    Because men propose in their preference order, we have $w \succ_{m'} w'$.
    However, this means $\mu$ is not stable.

    Thus, no woman can ever reject a stable partner in $MPDA$.
  \end{proof}

  By the previous claim, $MPDA$ moves the men down their preference
  lists the minimal amount required to enforce stability.
  Interestingly, a completely dual phenomenon occurs for the women's preferences.

  \begin{corollary}\label{claimMenBestStable}\label{claimWomenWorstStable}
    Let the set of men and women who receive a match at the end of $MPDA$
    be denote $\Mmatched$ and $\Wmatched$, respectively.
    In this matching $\mu$:
    \begin{enumerate} 
      \item every $m\in \Mmatched$ is paired to his best stable match.
      \item every $w\in \Wmatched$ is paired to their worst stable match.
    \end{enumerate}
  \end{corollary}
  \begin{proof}
    Over the course of $MPDA$, each man $m\in\Mmatched$ was rejected by every
    woman which he prefers to his partner in $MPDA$.
    By claim~\ref{claimRejectionUnstable}, this means his partner in $MPDA$ is
    his top stable match.

    Let $m\in \M$ and $w\in \W$ be paired by $MPDA$.
    Let $\mu$ be any stable matching which does not pair $m$ and $w$.
    We must have $w \succ_m \mu(m)$,
    because $w$ is the best stable partner of $m$.
    If $m \succ_w \mu(w)$, then $\mu$ is not stable.
    Thus, $w$ cannot be stably matched to any man she prefers less than $m$.
  \end{proof}

  The last claim also implies that
  the matching output by $MPDA$ is independent of the order in which
  men are selected to propose.

\subsection{General stable matchings}

  Interestingly, our claim~\ref{claimRejectionUnstable}, which related to
  $MPDA$, can be used to prove a fundamental property of the set of \emph{all}
  stable matchings.
  Specifically, we can prove the following weaker version of the rural
  hospital theorem\footnote{
    The full rural hospital theorem \cite{RothRuralHospital86}
    applies to many-to-one matching markets
    (i.e. the residents and hospitals problem).
    The conclusion is that if a hospital does not
    fill \emph{all} its openings in \emph{some} stable outcome,
    then it will fail to fill all its openings
    (and indeed receive exactly the same doctors) in
    \emph{every} stable outcome.
  } which will be key for much of our discussion in
  section~\ref{sectionLattice}.
  \begin{claim}[Rural Hospital Theorem] \label{claimRuralDoctors}
    Then the set of unmatched agents is the
    same across every stable outcome.
    % If a set of men $\overline \M$ are rejected by every woman during MPDA,
    % then no stable matching will match any man in $\overline \M$.
    % Moreover, in every stable matching, the set of unmatched men is the same.
  \end{claim}
  \begin{proof}
    Let $\Msingle$ be the set of men unmatched at the end of MPDA.
    Observe that each man in $\Msingle$ has proposed to every acceptable
    partner he has over the run of MPDA. Thus,
    claim~\ref{claimRejectionUnstable} implies that $\Msingle$ is unmatched
    in every stable outcome.
    On the other hand, reversing the roll of men and women and considering
    women-proposing deferred acceptance, we can see that the set of
    (un)matched women is also identical across every stable outcome.
  \end{proof}

  \newcommand{\graphwidth}{2.5}
  \newcommand{\graphheight}{1.5}
  \newcommand{\radius}{0.3}
    % claim 2.5 graph 
  \begin{wrapfigure}{rh}{3cm}
    \begin{tikzpicture}
      \draw[ultra thick] (0, \graphheight) -- (\graphwidth, \graphheight) node[pos=0.7, above] {$\mu$};
      \draw[thick, dashed] (0, \graphheight) -- (\graphwidth, 2*\graphheight) node[midway, above] {$\mu'$};
      \draw[thick, dashed] (0, 0) -- (\graphwidth, \graphheight) node[midway, below] {$\mu'$};
      \filldraw[fill=white] (0,0) circle[radius=\radius] node {$m'$};
      \filldraw[fill=white] (0,\graphheight) circle[radius=\radius] node {$m$};
      \filldraw[fill=white] (\graphwidth,\graphheight) circle[radius=\radius] node {$w$};
      \filldraw[fill=white] (\graphwidth,  2*\graphheight) circle[radius=\radius] node {$w'$};
    \end{tikzpicture}
  \end{wrapfigure}

  Claim~\ref{claimWomenWorstStable} seems to indicate that
  that the incentives of women and men are exactly opposite
  with regards to the results of man-proposing or women-proposing deferred
  acceptance.
  These next two claims prove that this is true for \emph{all} stable matchings.
  In \ref{subsectionLatticeIntro}, we investigate these ``order theoretic''
  properties further.

  \begin{claim}\label{claimOneUpOneDown}
    Let $\mu, \mu'$ be stable matchings, and say $\mu(m) = w$, but $\mu'(m)\ne w$.
    Then $\mu'(m) \succ_m w$ if and only if $\mu'(w) \prec_w m$.
  \end{claim}
  \begin{proof}
    $(\Leftarrow)$ ``If $w$ downgrades, then $m$ upgrades''.
    Suppose $\mu'(w) \prec_w m$. Because $\mu'$ is stable, yet $m$ and $w$
    are not matched in $\mu'$, we must have $\mu'(m) \succ_m w$,
    or else $(m,w)$ would form a blocking pair.
    (A rephrasing: this direction is easy because the definition of stability
    immediately makes it impossible for $m$ and $w$ to both downgrade).

    $(\Rightarrow)$ ``If $w$ upgrades, then $m$ downgrades''.
    Let $m' = \mu'(w) \ne m$ and $w' = \mu'(m) \ne w$.
    Suppose that $m' \succ_w m$, and for contradiction suppose that $w' \succ_m w$.
    Because $\mu'$ is stable, $(m', w')$ is not a blocking pair,
    so either $w\succ_{m'} w'$ or $m\succ_{w'} m'$.
    In the first case, $(m',w)$ form a blocking pair in $\mu$,
    and in the second case, $(m,w')$ form a blocking pair in $\mu$.
    Thus, in either case $\mu$ is not stable.
  \end{proof}
  \begin{claim}\label{claimDominateOposites}
    Let $\mu$ and $\mu'$ be stable matchings.
    Every man (weakly) prefers their match in $\mu$ over $\mu'$ if and only if
    every woman (weakly) prefers their match in $\mu'$ over $\mu$.
  \end{claim}
  \begin{proof}
    Suppose each $m\in\M$ has $\mu'(m)\succeq_m\mu(m)$.
    For each $w\in\W$ with $\mu(w)\ne\mu'(w)$, we must have
    $\mu'(w)\prec_w \mu(w)$ by claim~\ref{claimOneUpOneDown}.
    The proof for the other direction is identical.
  \end{proof}

%% file: LatticeStructure.tex
  Given that men and women have strictly opposite incentives, it is natural to
  define a dominance relationship over all stable matchings according to the
  preferences of one side of the market.

  \begin{definition}
    We say that a stable matching $\mu$
    \emph{woman-dominates} $\mu'$, written $\mu \ge \mu'$, if
    for every $w\in\W$, we have $\mu(w)\succeq_w \mu'(w)$
    (that is, every woman is at least as happy with her match in $\mu$ as in
    $\mu'$).
    For some fixed set of preferences,
    we let $\L$ denote the set of stable matchings
    of $P$, ordered by the relation $\ge$.
  \end{definition}
  We can define man-dominance analogously, and by
  claim~\ref{claimDominateOposites}, $\mu$ man-dominates $\mu'$ if and only if
  $\mu \le \mu'$.
  Now, one can visualize
  the collection of all stable matchings as starting with the unique man-optimal
  outcome at the bottom, the unique woman-optimal outcome at the top, and all
  other stable matching in between.
  % But what does this partial order graph
  % actually look like? Does it have any special properties? 

  In this section we show that the set of all stable
  matchings for a set of preferences $P$ forms what's called a
  \emph{distributive lattice} under the women-dominance order.
  % Then we give more intuition on the structure of the
  % stable matching lattice by showing the connection between the lattice and the
  % $MPDA$ algorithm.
  For the sake of completeness, we first discuss the relevant definitions.
  Informally, a lattice is a partial order in which, for any two elements $a,b$,
  there is a unique ``lowest element above $a$ and $b$'' (the join) and a
  ``highest element below $a$ and $b$'' (the meet)\footnote{
    Note that it follows from the definition that join and meet operations,
    if they exist, are unique.
  }.

  \begin{definition}
    A \emph{partial order} $\le$ is a reflexive, transitive, antisymmetric relation.
    We write $a < b$ when $a \le b$ and $a \ne b$.

    For elements $a,b$ of a partial order, a \emph{least upper bound} $a\vee b$ is
    an element such that $a\le a\vee b$ and $b\le a \vee b$, and for any 
    $c$ such that $a\le c$ and $b\le c$, we have $a\vee b\le c$.
    A \emph{greatest lower bound} $a\wedge b$ is defined analogously, interchanging $\le$
    with $\ge$.
    We also call $a\vee b$ the \emph{join} of $a$ and $b$ and $a\wedge b$ the
    \emph{meet} of $a$
    and $b$.

    A \emph{lattice} $L$ is a partial order in which there exist greatest lower
    bounds and least upper bounds for any $a,b\in L$.

    A \emph{chain} in a lattice is any ``totally ordered'' sequence
    $a_1 \le a_2 \le \ldots \le a_k$.

    A lattice $L$ is \emph{distributive} if the join and meet operations satisfy
    the following equations:
    \[ a \wedge (b \vee c) = (a\wedge b)\vee (a\wedge c) \]
    \[ a \vee (b \wedge c) = (a\vee b)\wedge (a\vee c) \]
  \end{definition}

  % Denote the woman dominance order as follows: If $\mu$ woman-dominates $\mu'$
  % we write $\mu\ge\mu'$. Now, we are ready to prove that under woman dominance
  % order $\ge$, the collection of all stable matching form a distributive
  % lattice. 

  The join and meet operations in $\L$ are very natural:
  the join of $\mu$ and $\mu'$ corresponds to the matching $\tilde\mu$ where
  each woman gets the better of her two partners from $\mu$ and $\mu'$.
  This is exactly the operation one would hope would work --
  clearly $\tilde\mu$ is the worst matching (for the women)
  in which women do at least as well as in $\mu$ and $\mu'$.
  % clearly $\mu, \mu' \le \tilde\mu$ and whenever $\mu, \mu' \le \mu^*$,
  % we have $\tilde\mu \le \mu^*$.
  We prove below that this operation always yields a stable matching.

  \begin{definition}
    Given stable matchings $\mu$ and $\mu'$, define 
    $\mu\vee\mu'$ such that, for each woman $w$,
    $(\mu\vee\mu')(w)$ is the most preferred partner
    of $w$ among $\mu(w)$ and $\mu'(w)$.
    Similarly, define
    $\mu\wedge\mu'$ such that each woman is matched to their least preferred
    partner from $\mu$ or $\mu'$.
    \end{definition}
  % With these operations spelled out,
  % we are ready to prove that $\L$ is a distributive lattice.

  \begin{theorem} \label{theoremDistributiveLattice}
    The collection $\L$ of all stable matchings of some instance
    form a distributive lattice under the dominance ordering $\le$,
    with join and meet given by $\vee$ and $\wedge$.
    % Moreover, the join and meet operations are given by $\vee$ and $\wedge$ as
    % defined above.
  \end{theorem}
  \begin{proof}
    It's easy to see that $\le$ forms a partial order on $\L$.
    We'll show that $\vee$ gives least upper bounds in $\L$.
    It's easy to see that, if $\tilde\mu = \mu\vee\mu'$ is a stable matching,
    then it must be the least upper bound of $\mu$ and $\mu'$.

    First, we claim that $\tilde\mu$ is a matching.
    Suppose for contradiction that some man $m$
    is the match of two women $w$ and $w'$ in
    $\tilde\mu$. Without loss of generality suppose $\mu(w)=m$,
    so $m=\mu(w)\succ_w \mu'(w)$,
    and $\mu'(w')=m$, so $m=\mu'(w')\succ_{w'}\mu(w')$.
    Applying claim~\ref{claimOneUpOneDown} twice,
    we get that $w=\mu(m)\prec_m \mu'(m)=w'$
    and also that $w'=\mu'(m)\prec_m \mu(m)=w$,
    a contradiction.

    Second, we claim that $\tilde\mu$ is stable.
    Suppose that $(m,w)$ is a blocking pair for $\tilde\mu$,
    Certainly the partners of $m$ and $w$ must be from different matchings among
    $\mu$ or $\mu'$, say $\tilde\mu(m)=\mu'(m)$ and 
    $\tilde\mu(w)=\mu(w)\ne \mu'(w)$.
    As $(m,w)$ is blocking, $w\succ_m\mu'(m)$ and $m\succ_w\mu(w)$.
    But by the definition of $\tilde\mu$, we have $\mu(w)\succ_w\mu'(w)$,
    so $m\succ_w\mu'(w)$ as well, and $\mu'$ is not stable.

    Now we show that $\wedge$ gives the greatest lower bound in $\L$.
    % For each man $m$, I 
    By claim~\ref{claimOneUpOneDown}, this is equivalent to defining
    $\mu\wedge\mu'$ such that every man
    gets their best partner from $\mu$ or $\mu'$
    (because $m = \mu(w) \prec_w \mu'(w)$ if and only if
    $w = \mu(m) \succ_m \mu'(m)$).
    Thus, the proof is identical to the proof given for $\vee$,
    interchanging men with women.

    % Finally, we sketch the proof that the join and meet operations 
    % In $\L$ are distributive.
    % Analogously to the above, we would define $\mu\wedge\mu'$ such that every
    % man gets their preferred partner from $\mu$ or $\mu'$.
    Finally, the join and meet operations are distributive for the same reason that
    the operations of min and max distribute over each other.
    In particular, we can fix a woman $w$ and see that (with max and min taken
    according to $\succ_w$)
    \[ \big(\mu_1 \wedge (\mu_2 \vee \mu_3) \big)(w)
      = \min \big\{ \mu_1(w), \max\big\{\mu_2(w), \mu_3(w)\big\} \big\} \]
    \[ \max \big\{ \min \big\{ \mu_1(w), \mu_2(w) \big\},
      \min \big\{\mu_1(w), \mu_3(w)\big\} \big\}
      = \big( (\mu_1 \wedge \mu_2) \vee (\mu_1 \wedge \mu_3) \big)(w) \]
    % Thus, for every woman $w$,
    % equal to the worse of $\mu_1(w)$ and (the better of $\mu_2(w)$ and
    % $\mu_3(w)$).
    % On the other hand, we have the better of
    % (the worse of $\mu_1(w)$ and $\mu_2(w)$)
    % and (the worse of $\mu_1(w)$ and $\mu_3(w)$),
    % which with a little thought you can realize that these two equations are
    % equal.
  \end{proof}

  % Now, one may ask why it is interesting that the collection of stable matchins
  % form a distributive lattice. It turns out distributive lattices has the
  % following interesting property called Birkhoff's representation theorem.  
  % \begin{theorem}[Birkhoff's representation theorem]
  %   For any finite distributive lattice $L$, there exists a subset $P$ of $L$
  %   (namely, the ``join irreducible elements'' of $L$, which is the set of all $c$ where for all $a, b$ such that $a\vee b = c$, we have $a=c$ or $b=c$) such that
  %   $L$ is isomorphic to the collection of downward-closed subsets of $P$
  %   (under the partial order induced by restricting $L$ to $P$), which
  %   forms a lattice under set containment.
  % \end{theorem}

  % The representation theorem gives us reason to believe that a compact
  % representation of the collection of stable matching $\L$ exist. Indeed, we
  % will show in section~\ref{sectionRotation} that there is a natural algorithmic
  % interpretation of the join irreducible elements of $\L$. 

  % We will not directly use this theorem, but it serves as motivation for why one might expect a compact representation of $\L$ to exist.
  % In general, when one encounters a distributive lattice, it's always useful to ask what its join irreducible elements are, and see if there's a natural mathematical or algorithmic interpretation.

  The most important lattice-theoretic concept we will need is the notion of
  \emph{covering}. Informally, an element covers another in a lattice if there
  is no element between them in the ordering.

  \begin{definition}
    For $a, b$ elements of a lattice, we say $a$ \emph{covers} $b$
    % , denoted $a \gtrdot b$,
    when $a > b$ and no element $c$ exists with $a > c > b$.
  \end{definition}

  A useful equivalent definition of covering is the following:
  $a$ covers $b$ if and only if whenever $a \ge c > b$, we have $a = c$.
  Although the concept of covering relations is central to our paper,
  we need remarkably few formal properties of covering relations
  (or of lattices for that matter).
  Here is what we will need:
  \begin{claim}
    In any finite lattice and
      % \item If $a \le b$ and $a \le c$, then $a \le b \wedge c$.
      %   Conversely, if $a \ge b$ and $a \ge c$, then $a \ge b \vee c$.
    for any $a \le b$, there exists a sequence
    $a = a_0 < a_1 < \ldots < a_k = b$ (for some $k\ge 0$)
    such that $a_i$ covers $a_{i-1}$ for each $i$.
    Such a sequence is called a \emph{maximal chain} between $a$ and $b$.
  \end{claim}
  \begin{proof}
    If $a = b$ we are done. Otherwise, let $S$ be the set of
    all elements $c$ such that $a < c < b$,
    and induct on $|S|$.
    If $|S|=0$, then $b$ covers $a$ and we are done.
    Otherwise, take any $c\in S$.
    Consider the set of all $d$ such that $a < d < c$.
    For such a $d$, we have $a < d < b$ and also $d \ne c$.
    Thus, there are strictly fewer than $|S|$ such $d$.
    Thus, by induction, there exists a maximal chain between $a$ and $c$.
    Similarly, there exists a maximal chain between $c$ and $d$,
    so the concatenation of these two chains gives us a maximal
    chain between $a$ and $b$.

    % Given any $a < b$, if no $c$ exists with $a < c < b$, then we are done.
    % Otherwise, consider $a < c$ and $c < b$. If any $d$ exists between either of
    % these pairs, say $a < d < c$, then consider $a < d < c < b$.
    % Condtinuing in this way, we always have a chain of distinct elements
    % between $a$ and $b$.
    % As the lattice is finite, this process cannot continue indefinitely,
    % so we must reach a chain where every adjacent element is in a covering
    % relation.
  \end{proof}

  We are interested in covering relations in $\L$ because they describe
  the ``minimal differences'' needed to go from one stable matching to another.
  The previous claim hints that one can describe any matching $\mu$ by giving
  the ``covering relations leading up to $\mu$''.
  Eventually, we will describe all covering relations (using ``rotations'') and
  show how you can represent all stable matchings as certain subsets of these
  ``minimal differences''.

%% file: LatticeExploration.tex
In this section, we study how the lattice-theoretic properties of $\L$ start to
manifest algorithmically in certain special cases of $MPDA$.
We'll characterize the covering relations (and thus the entire structure of the
lattice), essentially in terms of execution traces of $MPDA$.
Intuitively, the main result is that, starting from any stable matching,
if a woman has a better stable partner, then she can ``divorce'' her husband,
and if we keep running $MPDA$, we will arrive at a stable matching preferred by
that woman.

Consider a fixed set of input preferences $P$.
\begin{definition}
  For a set of preferences $P$ and matching $\mu$ stable under $P$, define
  $P(\mu)$ as follows: every woman $w$ matched in $\mu$ truncates 
  the end of their preference
  list just after $\mu(w)$ (removing all men ranked worse than
  their current match), and every man $m$ matched in $\mu$
  truncates the beginning their preference list just before $\mu(m)$
  (removing all women ranked better than their current match).
  Women unmated in $\mu$ keep their full preference list,
  and men unmatched in $\mu$ are removed from $P(\mu)$.

  For a woman $w$ matched in $\mu$,
  define $P_w(\mu)$ the same as $P(\mu)$ with one additional change:
  woman $w$ truncates her preference list one more place
  by removing her current match $\mu(m)$.
\end{definition}

Intuitively, $P(\mu)$ defines the state we are in after a deferred acceptance
type algorithm reaching matching $\mu$: the women are still seeking to improve
beyond their current match and the men are still proposing down their lists.
On the other hand, $P_w(\mu)$ represents preferences corresponding to woman $w$
attempting to reach next to a better match than $\mu(w)$
(by rejecting $\mu(w)$ in $MPDA$).

In what follows, we call a match stable if it is stable for the original set
of preferences $P$. If we need to refer to the fact that a match is stable for
some truncated set of preferences $P_w(\mu)$, we will specify so.
We denote $\Mmatched$ and $\Msingle$ as the set of men who are matched
and unmatched respectively in the stable matchings with the original
preferences $P$ (recall from \ref{claimRuralDoctors} that these sets are
uniquely determined). Define $\Wmatched$ and $\Wsingle$ analogously.

Note that the execution of $MPDA(P_w(\mu_0))$ is quite a bit 
more simple than a general
execution of $MPDA$. After the first proposal of each man in $\Mmatched$
(i.e. each $m\in \Mmatched \setminus \{\mu_0(w)\}$ proposes to and is accepted
by $\mu_0(m)$), there is exactly one ``free'' man from $\Mmatched$
at a time (i.e. one man who is not tentatively matched and
still proposing down his list),
until there is no longer a free man and the execution terminates.
Specifically, the free man is initially $\mu_0(w)$, and if a proposal from
the free man is accepted by a woman $w'\in \Wmatched\setminus\{w\}$,
the free man becomes $\mu_0(w')$.
If a proposal is accepted by $w$ or a woman from $\Wsingle$,
or if a man proposes to the last woman on his preference list,
the algorithm terminates.
In order to capture such an execution sequence
we make the following definition:
\begin{definition}
  Given a stable matching $\mu_0$ and a woman $w\in\Wmatched$,
  the \emph{rejection chain of $w$ starting from $\mu_0$}
  is the list $(w_1, m_1, w_2, m_2, \ldots, a_i)$ defined as follows:
  \begin{itemize}
    \item $w_1 = w$ and $m_1 = \mu_0(w)$
    \item The men $m_i$ are, in order, the men from $\Mmatched$ which are
      free during the execution of $MPDA(P_w(\mu_0))$
    \item For each $i$, $w_{i+1}$ is the woman (if any)
      who accepts a proposal from $m_i$
    \item The list ends when the algorithm terminates
  \end{itemize}

  We also call such a list ``the rejection chain of $MPDA(P_w(\mu_0))$''
  or just ``the rejection chain'' if $w$ and $\mu_0$ are understood.
\end{definition}

We start by establishing some basic properties relating rejection chains
to the match returned by $MPDA(P_w(\mu_0))$. The proof is immediate.
\begin{claim}\label{claimRejectionChainEquiv}
  Let $\mu_0$ be a stable match and let $w\in\Wmatched$.
  Let $\mu'$ be the result of $MPDA(P_w(\mu_0))$ and
  let $(w_1, m_1, w_2, m_2, \ldots, a_i)$ be the rejection chain
  of $w$ starting from $\mu_0$ (so $a_i$ denote the last agent
  in the rejection chain).
  Then exactly one of the following is true:
  \begin{itemize}
    \item $a_i \in \Wsingle$ is a woman who is now matched in $\mu'$
    \item $a_i$ is a man from $\Mmatched$ who is now unmatched in $\mu'$
    \item $a_i = w$ (and $w$ receives a match in $\mu'$ if and only if
      $a_i=w$)
  \end{itemize}
  Moreover, the set of agents matched in $\mu'$ is the same as that
  in $\mu_0$ if and only if $a_i = w$.
\end{claim}

Our goal is to expore the stable matching lattice $\L$ using
the operation $(\mu_0, w) \mapsto MPDA(P_w(\mu_0))$.
Thus, the first thing we need to know is when this operation keeps us in the
lattice $\L$ and when the result is an unstable matching.

\begin{claim}\label{claimRejectAndMoveToStable}
  Let $\mu_0$ be stable and take $w\in\Wmatched$.
  Let $MPDA(P_w(\mu_0))$ terminate in a matching $\mu'$.
  Then $\mu'$ is stable if and only if
  $w$ receives a match in $\mu'$.
\end{claim}
\begin{proof}
  If $w$ is not matched in $\mu'$, then
  the set of matched agents differs between $\mu_0$ and $\mu'$.
  Thus $\mu'$ cannot possibly be stable by the rural hospital
  theorem~\ref{claimRuralDoctors}.

  On the other hand, suppose $w$ receives a match in $\mu'$.
  By the previous claim, this means 
  % the rejection chain of $w$ starting from
  % $\mu_0$ is of the form $(w, m_1, w_2, m_2, \ldots, m_i, w)$, and
  that the set of agents matched in $\mu_0$ and $\mu'$ are identical.
  Consider how $MPDA(P_w(\mu_0))$ runs, converting from $\mu_0$ to $\mu'$.
  % Observe that, as $MPDA(P_w(\mu_0))$ runs, each man can only go down in their
  % preference from $\mu_0$ to $\mu'$.
  Observe that, because $w$ receives a match (which she prefers to
  $\mu_0(w)$), every woman can only improve their preference for their match.

  For the sake of contradiction, suppose $(m',w')$ is a blocking pair in
  $\mu'$. Certainly, $\mu'$ is stable for the preferences $P_w(\mu_0)$.
  How can $(m',w')$ be a blocking pair in $P$ but not in $P_w(\mu_0)$?
  The only way is if one agent truncated the other off their preference list in
  $P_w(\mu_0)$. We have two cases:
  \begin{enumerate}
    \item Suppose $w'$ truncated $m'$. % , but $m'$ did not truncate $w'$.
      Then we have $m'\preceq_{w'} \mu_0(w') \preceq_{w'} \mu'(w')$.
      % $w'\preceq_{m'}\mu_0(m')$.
      But $m' \succ_{w'} \mu'(w')$, a contradiction.
    \item Now suppose $m'$ truncated $w'$. Then $w' \succ_{m'} \mu_0(m')$.
      But then $m' \succ_{w'} \mu'(w') \succeq_{w'} \mu_0(w')$,
      so $(m',w')$ are unstable in $\mu_0$, a contradiction.
  \end{enumerate}

  % If neither of $m'$ or $w'$ appear in the rejection chain,
  % then they receive the same match in $\mu'$ as in $\mu_0$ (or go unmatched in
  % both). Because $\mu_0$ is stable, this is not possible.
  % Thus, we have two cases:
  % \begin{enumerate}
  %   \item Suppose $w' = w_i$ for some $i$, but $m'$ does not change his match
  %     from $\mu_0$ to $\mu'$.
  %     We have $w' \succ_{m'} \mu'(m') = \mu_0(m')$
  %     and $m' \succ_{w'} \mu'(w')\succ_{w'}\mu_0(w')$ (as $w'$ can only go up
  %     from $\mu_0$ to $\mu'$).
  %     Thus we get that $\mu_0$ was unstable, a contradiction.

  %   \item Suppose $m' = m_i$ for some $i$.
  %     Sense $\mu_0$ was stable and $m'\succ_{w'} \mu'(w') \succeq_{w'} \mu_0(w')$,
  %     we must have $w'\prec_{m'} \mu_0(m')$. Thus $w'$ is still
  %     on the truncated preference list of $m'$.
  %     So $m'$ would propose to $w'$
  %     before his match in $\mu'$. As $m'\succ_{w'}\mu'(w')$, this
  %     means $w'$ should get a match at least as good as $m'$, a
  %     contradiction.
  % \end{enumerate}

  % For the sake of contradiction, suppose $(m',w')$ is a blocking pair in
  % $\mu'$ for preferences $P$. We know that $\mu'$ is stable for preference
  % $P_w(\mu_0)$. Thus, one of $m'$ or $w'$ must have removed the other from their
  % preference list in $P_w(\mu_0)$.
  % First, suppose $w'$ removed $m'$ from her list,
  % but As $w'$ received a better
  % match than $\mu_0(w')$
\end{proof}

Next, we need to know that, if we have not reached the woman-optimal stable
match, then we can \emph{always keep moving up in the lattice}.
Intuitively, this is true because, whenever
a stable matching exists, $MPDA$ will find it,
so if a stable matching with women receiving good partners
exists, then $MPDA$ will find it as well.

\begin{claim}\label{claimCanGoUpIfNotWosm}\label{subclaimWomanOptPartner}
  If $\mu_0$ is a stable matching in which $w$ is not paired to her optimal stable
  partner, then $MPDA(P_w(\mu_0))$ will return a stable matching $\mu'$
  which strictly woman-dominates $\mu_0$, i.e. $\mu' > \mu_0$.

  Conversely, if $MPDA(P_w(\mu_0))$ fails to return a stable match, then
  $w$ is matched to her optimal stable partner in $\mu_0$.

  Moreover, if $\mu'$ covers $\mu_0$ in $\L$, then
  $MPDA(P_w(\mu_0))$ returns $\mu'$ for \emph{any} woman $w$ who 
  receives a better partner in $\mu'$ than in $\mu_0$.
  % . preferred by $w$.
\end{claim}
\begin{proof}
  Let $\mu^*$ be any stable matching in which $w$ has a better partner
  than in $\mu_0$, i.e. $\mu^*(w)\succ_w\mu_0(w)$.
  Without loss of generality, we can assume that $\mu^*\ge \mu_0$,
  because if this is not the case, we can replace $\mu^*$ with 
  $\mu^* \vee \mu_0$.
  By the rural hospital theorem (claim~\ref{claimRuralDoctors}),
  $\mu^*$ must have exactly the same set of matched agents as in $\mu_0$.

  Let $\mu' = MPDA(P_w(\mu_0))$,  and note that $\mu'$ is certainly stable for
  preferences $P_w(\mu_0)$.
  Because $\mu^*\ge\mu_0$ and $w$ gets matched strictly above
  $\mu_0(w)$, the matching $\mu^*$ is also stable for preferences $P_w(\mu_0)$.
  Thus, once again $\mu^*$ and $\mu_0$ have identical sets of matched agents. 
  By claim~\ref{claimRejectAndMoveToStable}, we conclude that $\mu'$ is stable
  for preferences $P$. 

  % Thus, by the fact that $MPDA$ returns the man-optimal stable outcome
  % (claim~\ref{claimMenBestStable}), the result of $MPDA(P_w(\mu_0))$
  % will be the matching $\mu'$ with the following property:
  % $\mu'$ is the man-optimal outcome which is stable for preferences
  % $P_w(\mu_0)$. 

  % Since $\mu^*, \mu$ are both stable for preferences $P$,
  % the set of matched agents in $\mu^*$ and $\mu$ must be the same. Similarly,

  % But we just saw that if 1) every matched man receives a match,
  % and 2) only the previously-matched women are those matched,
  % then $\mu$ will be stable for the original preferences $P$.
  % But because $\mu'$ is stable for preferences $P_w(\mu_0)$, every man in
  % $\mu_0$ is matched above every woman unmatched in $\mu_0$.

  By the definition of $P_w(\mu_0)$,
  each woman will only accept a proposal in $MPDA(P_w(\mu_0))$
  from a man she likes at least as much as in $\mu_0$.
  As $w$ receives a strictly better
  match in $\mu'$, we have $\mu' > \mu_0$.

  For the converse, suppose $w$ is matched to her optimal stable partner in
  $\mu_0$. In $\mu = MPDA(P_w(\mu_0))$, $w$ will not accept a proposal except
  from a man ranked above her match in $\mu_0$. Thus, $\mu$ cannot possibly be
  stable, as then $w$ would be matched in $\mu$ to a stable partner better than
  $\mu_0(w)$.

  Now, suppose $\mu'$ covers $\mu_0$, so that whenever
  $\mu' \ge \mu > \mu_0$, we have that $\mu = \mu'$,
  and let $w$ be any woman recieving a better match in $\mu'$ than in $\mu_0$.
  We claim that $\mu'$ is the man-optimal stable outcome in which
  each woman in $\Wmatched$ receives a partner at least as good as in $\mu_0$,
  and where $w$ receives a strictly better partner.
  Indeed, if $\mu'$ were \emph{not} this matching, then some $\mu$ would exist
  such that $\mu' > \mu > \mu_0$, and so $\mu'$ would not cover $\mu_0$.
  % a contradiction.
  By the fact that MPDA returns the man-optimal stable outcome
  (claim~\ref{claimMenBestStable})
  this exactly menas that $\mu'$ is the result of $MPDA(P_w(\mu_0))$.

\end{proof}

The previous claims show that if a woman $w$ has a better stable match,
she can reject her current match,
and if we continue running deferred acceptance then $w$ will achieve a better
outcome.
% and $w$ receives a better match, then we have found a new stable outcome.
Next we get a characterization of when these changes
from matching to matching are as small as possible 
(i.e. when the new matching covers the old in the lattice $\L$).

\begin{claim} \label{claimCoveringCondition}
  Suppose $\mu_0$ is stable and
  $MPDA(P_w(\mu_0))$ terminates in a stable matching $\mu' > \mu_0$.
  % Let $E = (w, m_1, w_2, m_2, \ldots, w_k, m_k, w)$
  % be the rejection chain of $w$ starting from $\mu_0$.
  Then $\mu'$ covers $\mu_0$ if and only if
  during the run of $MPDA(P_w(\mu_0))$, no woman from $\Wmatched$
  receives more than one proposal from men who she strictly
  prefers to her match in $\mu_0$.
  % the rejection chain
  % of $w$ starting from $\mu_0$
  % contains no agent other than $w$ more than once.

  % MIGHT WANT TO REWRITE THIS TO BE AN UNCONDITIONAL IF AND ONLY IF.
\end{claim}
\begin{proof}
  For this proof, call a proposal \emph{good} if it is made 
  by some man $m'$ to some woman
  $w'$, where $w'$ prefers $m'$ to $\mu_0(w')$.
  Note that a woman does not necessarily accept a good proposal
  (if she has already seen a proposal from a man she likes even more).

  ($\Rightarrow$) 
  Suppose that, while running $MPDA(P_w(\mu_0))$, some woman
  sees more than one good proposal.
  Because $MPDA(P_w(\mu_0))$ terminates as soon as $w$ sees a good proposal,
  this woman cannot be $w$.
  Let $w^*\ne w$ be the \emph{first} such woman,
  i.e. when $w^*$ receives her second good proposal,
  no other woman has yet received a second good proposal.

  Consider running $MPDA(P_{w^*}(\mu_0))$, and call the result $\mu$.
  As $MPDA$ progresses, we know that each woman receives exactly one good
  proposal, because $w^*$ was the first instance where a woman received two
  good proposals. Thus, the rejection chain of $MPDA(P_{w^*}(\mu_0))$ is
  a sublist of the rejection chain of $MPDA(P_w(\mu_0))$,
  with one notable exception: $w^*$ might not accept her second good proposal
  in $MPDA(P_w(\mu_0))$, but she will definitely accept the corresponding
  proposal in $MPDA(P_{w^*}(\mu_0))$.
  Regardless of this event, every woman who changes partners in
  $MPDA(P_{w^*}(\mu_0))$ will also change partners in 
  $MPDA(P_w(\mu_0))$, and indeed will do at least as well in the end,
  so $\mu \le \mu'$.
  As $w$ could not have possibly changed partners in $MPDA(P_{w^*}(\mu_0))$,
  this means $\mu < \mu'$. We already knew that $\mu_0 < \mu$,
  so this completes the proof that $\mu'$ does not cover $\mu_0$.

  ($\Leftarrow$) 
  For the other direction, suppose no woman 
  sees multiple good proposals.
  Now suppose $\mu_0 < \mu \le \mu'$ for some stable match $\mu$.
  Let $E = (w, m_1, w_2, m_2, \ldots, w_k, m_k, w)$ be the rejection
  chain of $MPDA(P_w(\mu_0))$.
  We'll show that, for any $i$, the outcome of $MPDA(P_{w_i}(\mu_0))$ is also $\mu'$.
  For each woman $w_i$ in $E$,
  consider the rejection chain $E_i$ of $w_i$ starting at $\mu_0$.
  % During the execution of $MPDA(P_{w_i}(\mu_0))$,
  % each woman in $E$ is either matched
  % as in $\mu_0$ or as in $\mu'$.
  For each man $m_j$, consider each woman $w$ on his preference
  list strictly between $w_j$ and $w_{j+1}$.
  Each such woman rejected him in $MPDA(P_w(\mu_0))$,
  and the only way for a woman to then accept $m_j$ in 
  $MPDA(P_{w_i}(\mu_0))$ is if $m_j$ is a good proposal for $w$,
  but $w$ had already seen an (even better) good proposal in $MPDA(P_w(\mu_0))$.
  Because we assume no woman receives multiple good proposals, this is
  impossible, so each $w$ between $w_j$ and $w_{j+1}$ will still reject $m_j$.
  Furthermore, $w_{j+1}$ will still accept him, as she has not yet seen a good
  proposal when $m_j$ proposes to her.
  Thus, each link of $E_i$ will be the same as in $E$,
  that is, $E_i$ will simply be
  $(w_i, m_i, w_{i+1}, \ldots, m_{i-1}, w_i)$ (with indices taken mod $k$).
  Thus, the outcome of $MPDA(P_{w_i}(\mu_0))$ is $\mu'$.

  Because $\mu_0 \ne \mu$, some woman
  must receive a strictly better match in $\mu$ than in $\mu_0$.
  As $\mu \le \mu'$, that woman must be $w_i$ for some $i$.
  Because $MPDA$ returns man-optimal stable outcomes
  (claim~\ref{claimMenBestStable}),
  $\mu' = MPDA(P_{w_i}(\mu_0))$ is the man-optimal stable
  outcome % (for preferences $P_{w_i}(\mu_0)$)
  in which every woman receives a match at least as good as in
  $\mu_0$, and in which $w_i$ receives a strictly better match.
  As $\mu$ is such a matching,
  we have $\mu' \le \mu$ and thus $\mu' = \mu$.
  Because $\mu'$ was an arbitrary element of $\L$ with $\mu_0< \mu\le\mu'$,
  we've shown that $\mu'$ covers $\mu_0$.
  % By claim~\ref{claimMenBestStable}, $\mu'$ is the man-optimal stable match
  % under preferences $P_w(\mu_0)$. Because all men receive matches worse than
  % their match in $\mu_0$ in both $\mu$ and $\mu'$, this means that $\mu'$
  % is the man-optimal stable match in which $w$ receives a strictly better
  % match than in $\mu_0$.
\end{proof}

\paragraph{Remark.} With the results of this section, we could already build
the entire stable matching lattice $\L$, represented by its covering
relations. Namely, we could essentially breadth-first search the lattice $\L$,
finding those matching which cover a given $\mu_0$ by calculating
$MPDA(P_w(\mu_0))$ for each woman $w$ (and keeping track of whether any woman
receives multiple proposals from a man she prefers to her old match from
$\mu_0$).

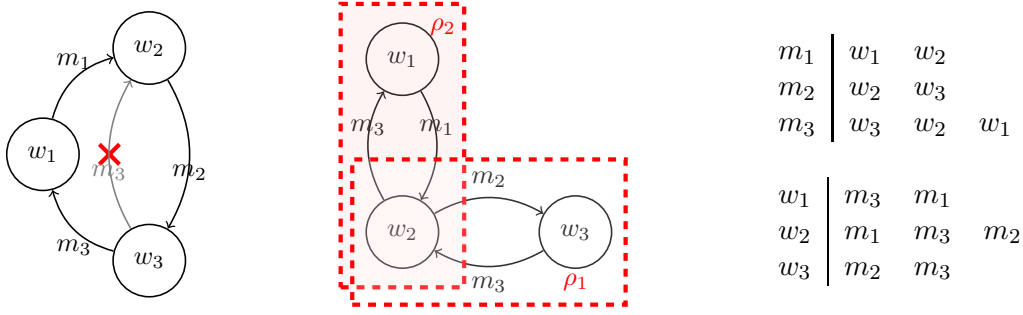
\begin{figure}[htp]
    \begin{tabular}[h]{m{1cm} m{4cm} m{5cm} m{4cm} }
    &
   \begin{tikzpicture}[node distance=2cm, semithick, auto, font=\small, cross/.style={cross out, draw, 
         minimum size=4*(#1-\pgflinewidth), 
         inner sep=0pt, outer sep=0pt}, cross/.default={3.5pt}] 
        \node[state, fill=white] (w1) {$w_1$};
        \node[state, fill=white] (w2) [above right of=w1] {$w_2$};
        \node[state, fill=white] (w3) [below right of=w1]
        {$w_3$};
        \node[ultra thick, cross, red] (x) at ($(w2)!0.5!(w3) + (-0.53, 0)$) {};
        \path[->] (w1) edge[bend left, above] node {$m_1$} (w2)
                  (w2) edge[bend left, below] node {$m_2$} (w3)
                  (w3) edge[bend left, below] node {$m_3$} (w1)
                  (w3) edge[bend left, below, gray] node {$m_3$} (w2); 
        
      \end{tikzpicture}
        & 
        \begin{tikzpicture}[node distance=2.3cm, semithick, auto, font=\small, cross/.style={cross out, draw, 
         minimum size=4*(#1-\pgflinewidth), 
         inner sep=0pt, outer sep=0pt}, cross/.default={3.5pt}] 
        \node[state, fill=white] (w1) {$w_1$};
        \node[state, fill=white] (w2) [below of=w1] {$w_2$};
        \node[state, fill=white] (w3) [right of=w2]
        {$w_3$};
        \coordinate (w2') at ($(w2.north) + (-1.4em, +0.9em)$);
        \coordinate (w3') at ($(w3.south) + (+1.4em, -0.9em)$);
        \coordinate (w1') at ($(w1.north) + (-1.8em, +0.3em)$);
        \coordinate (w2'') at ($(w2.south) + (+1.8em, -0.3em)$);
        \path[->] (w1) edge[bend left, above] node {$m_1$} (w2)
                  (w2) edge[bend left, above] node {$m_3$} (w1)
                  (w3) edge[bend left, below] node {$m_3$} (w2)
                  (w2) edge[bend left, above] node {$m_2$} (w3);
        \node [ultra thick, dashed, draw=red, fit= (w1') (w2''), fill=red!20, fill opacity=0.2] (r2) {};
        \node [ultra thick, dashed, draw=red, fit= (w2') (w3'), fill=white, fill opacity=0.2] (r1) {};
        
        \node [yshift=+2.0ex, xshift=+3.0em, red, font=\small] at (r1.south) (r1label) {$\rho_1$};
        \node [yshift=-2.0ex, xshift=+1.4em, red, font=\small] at (r2.north) (r2label) {$\rho_2$};
      \end{tikzpicture} & 
       \begin{tabular}{l}
       \begin{tabular}{c | c c  c }
        $m_1$ & $w_1$ & $w_2$ & \\
        $m_2$ & $w_2$ & $w_3$ & \\
        $m_3$ & $w_3$ & $w_2$ & $w_1$ \
      \end{tabular}\\ \\
      
      \begin{tabular}{c | c c c }
          $w_1$ & $m_3$ & $m_1$  \\
          $w_2$ & $m_1$ & $m_3$ & $m_2$\\
          $w_3$ & $m_2$ & $m_3$ 
      \end{tabular}
       \end{tabular}
  \end{tabular}
  \caption{A rejection chain with no repeated agents does not 
    imply a covering relationship}
  \label{fig:chainVSDepend}
  \end{figure}

  \paragraph{Example.}
  The condition in claim~\ref{claimCoveringCondition} is subtly different from
  an agent appearing multiple times in the rejection chain. For instance,
  consider the example illustrated in Figure~\ref{fig:chainVSDepend}, and let
  $\mu_0=\{(m_1, w_1), (m_2, w_2), (m_3, w_3)\}$ be the man-optimal stable
  outcome. Agents only appear one time in $MPDA(P_{w_1}(\mu_0))$'s rejection
  chain $(w_1, m_1, w_2, m_2, w_3, m_3, w_1)$. 
  However, $\mu^*:=\allowbreak \{(m_1, w_2),\allowbreak 
  (m_2, w_3),\allowbreak (m_3, w_1)\}$,
  the resulting stable matching from
  $MPDA(P_{w_1}(\mu_0))$, does not cover $\mu_0$.  In fact,
  $MPDA(P_{w_2}(\mu_0))$ results in a stable matching 
  $\mu_1 = \{(m_1, w_1), (m_2, w_3), (m_3, w_1)\}$
  such that $\mu_0 < \mu_1 < \mu^*$.
  % where the women are strictly worse off in $\mu_1$
  % comparing to $\mu^*$.
  Intuitively, what happened is that, because $w_2$ received
  a proposal from both $m_1$ and $m_3$ in $MPDA(P_{w_1}(\mu_0))$,
  she actually had \emph{two} opportunities to upgrade and reach a better stable
  matching.
  Thus, the change from $\mu_0$ to $\mu^*$ can be
  broken down into two steps, where one step must come before the
  other (namely, $\mu_1$ must be reached before $\mu^*$).
  In the next section,
  we'll see how to formalize these concepts using \emph{rotations} and
  \emph{predecessor relations}.

%% file: AbstractRotations.tex
We now define a concise way to describe the difference between
``consecutive'' stable matchings, i.e. pairs of matchings where one covers the other.
The collection of these ``minimal differences'' will allow us to represent all
stable matchings in a principled and compact way.
% Using the machinery developed in section~\ref{sectionLattice},
% it will be easy to prove that every covering relation corresponds to a
% rotation, as well as a wealth of other properties.

\begin{definition}
  Let $\mu\in\L$ be a stable matching and
  $\rho = [ (w_0,m_0), (w_1,m_1), \ldots, (w_{k-1}, m_{k-1}) ]$
  a list of agents
  with each $w_i\in\W$ and $m_i\in\M$, and $\mu(w_i) = m_i$ for each $i$.
  The \emph{elimination of $\rho$ from $\mu$} is the matching
  $\mu'$ such that $\mu'(m_i) = w_{i+1}$ for $i=0,\ldots,k-1$ (with indices
  taken mod $k$) and $\mu'(m) = \mu(m)$ for each $m$ which doesn't appear in
  $\rho$.

  We say $\rho$ is a \emph{rotation exposed in $\mu$}
  when $\mu'$ is a stable matching, and
  $\mu'$ \textbf{covers} $\mu$.

  The collection of all rotations which are exposed in some $\mu\in\L$
  is called the \emph{set of rotations}, and is denoted by $\Pi$.
  % \begin{itemize}
  %   \item $\mu(m_i) = w_i$ for each $i$
  %   \item If $\mu'$ is the matching such that $\mu'(m_i) = w_{i+1}$
  %     for each $i$ (with indices taken mod $k$),
  %     and $\mu'(m) = \mu(m)$ for $m$ which don't appear in $\rho$,
  %     then $\mu'$ is stable and $\mu'$ covers $\mu$ in $\L$
  % \end{itemize}
  % In this case, we call
  % $\mu'$ the \emph{elimination of $\rho$ from $\mu$}.
\end{definition}

We can ``visualize'' rotations as follows: if the men in the rotation all
``get up'' from their match in $\mu$ and move one place to the right
(cyclically) in the rotation, then we arrive at a new stable match
(which covers the old one).
Note that we only call a list of agents $\rho$ a rotation when
there exists a $\mu_0$ such that the elimination of $\rho$ from $\mu_0$
covers $\mu_0$.
We view two rotations as equivalent if they differ by a cyclic shift,
i.e. $\rho$ as above is identified with $[(w_i,m_i), (w_{i+1},m_{i+1}),
\ldots,(w_{i-1},m_{i-1})]$ (indices taken mod $k$) for any $i$.
By the definition, it is clear that such a shift changes nothing.
For the rest of this section, all indices in rotations are considered mod
$k$ where $k$ is the length of the rotation.
We say each pair $(w_i, m_i)$ \emph{appears} in rotation $\rho$,
and that $\rho$ \emph{moves $m_i$ from $w_i$ to $w_{i+1}$}
and \emph{moves $w_i$ from $m_i$ to $m_{i-1}$}.
If $\rho$ moves $m$ from $w_i$ to $w_{i+1}$,
and $m$ ranks $w$ between $w_i$ and $w_{i+1}$ (that is,
$w_i\succ_m w\succ_m w_{i+1}$), 
we say that $\rho$ moves $m$ from above $w$ to below $w$.
Define the meaning of the phrase
``$\rho$ moves woman $w$ from below $m$ to above $m$''
and related phrases analogously.

Given the discussion in the previous section, we arrive easily at a rich set
of claims characterizing rotations and their relationship to each other.
Claim~\ref{claimRotationFromMPDA} translates the language of rotations to the
concept of $MPDA$ with truncated lists, as discussed in
section~\ref{sectionLattice}, then claim~\ref{claimStablePartnerRotations} lists
the basic properties of rotations.

\begin{claim}\label{claimRotationFromMPDA}
  For any stable matching $\mu_0$,
  the following are equivalent:
  \begin{itemize}
    \item $\rho = [(w_0,m_0), \ldots, (w_{k-1},m_{k-1})]$
      is a rotation exposed in $\mu_0$,
      and $\mu'$ is the elimination of $\rho$ from $\mu_0$
    \item $MPDA(P_{w_0}(\mu_0))$ produces the stable matching $\mu'$, and
      during its execution no woman receives multiple proposals from a man
      she prefers to her match in $\mu_0$,
      and the rejection chain of $w_0$ starting from $\mu_0$ is exactly
      $(w_0, m_0, w_1, m_1, \ldots, w_{k-1}, m_{k-1}, w_0)$.
  \end{itemize}
  Moreover, $\mu'$ covers $\mu_0$ if and only if there exists a rotation $\rho$
  exposed in $\mu_0$ such that $\mu'$ is the elimination of $\rho$ from $\mu_0$.
\end{claim}
\begin{proof}
  By claim~\ref{claimCoveringCondition}, $\mu' = MPDA(P_{w_0}(\mu_0))$ 
  is a stable matching which covers $\mu_0$ if and only if during its execution
  no woman receives multiple proposals from a man she prefers to her match in
  $\mu_0$, and $w_0$ receives a match in $\mu'$.
  In this case, the rejection chain is of the form
  $(w_0, m_0, w_1, m_1, \ldots, w_{k-1}, m_{k-1}, w_0)$,
  and the stable matching $\mu'$ is exactly is exactly the
  elimination of $\rho =[(w_0,m_0), \ldots, (w_{k-1},m_{k-1})]$ from $\mu_0$.

  By definition, the elimination of a rotation from $\mu_0$ always covers $\mu_0$.
  Furthermore, claim~\ref{claimCanGoUpIfNotWosm} tells us that whenever $\mu'$
  covers $\mu_0$, running $MPDA(P_w(\mu_0))$ will produce $\mu'$ (and by
  claim~\ref{claimCoveringCondition} this will differ from $\mu_0$ by the
  elimination of a rotation).
\end{proof}

\begin{claim} \label{claimStablePartnerRotations}
  We have the following:
  \begin{enumerate}
    %\item \label{subclaimEachCoverHasRot}
    %  If $\mu'$ covers $\mu_0$, then there exists exactly one rotation
    %  exposed in $\mu_0$ such that $\mu'$ is the elimination of $\rho$
    %  from $\mu_0$.
    \item \label{subclaimStabPartners}
      $(w,m)$ are stable partners (i.e. matched in some stable matching)
      if and only if
      $(w,m)$ appears in some rotation in $\Pi$
      or $(w,m)$ are paired in the woman-optimal stable outcome.
    \item \label{subclaimNextStablePartner}
      Let $(w_i,m_i)$ appear in some rotation (indexed as above).
      Then $m_{i-1}$ is the worst-ranked stable partner of $w_i$
      who $w_i$ ranks above $m_i$
      (and similarly $w_{i+1}$ is the best-ranked stable partner of $m_i$
      who $m_i$ ranks below $w_i$). In other words, rotations move agents to
      their ``next'' stable partners (for women, the next best stable partner,
      and for men, the next-worst).
    \item \label{subclaimAtMostOnceRotate}
      A pair $(w,m)$ of men and women appear in at most one rotation together.
    \item \label{subclaimNumberRotations}
      There are at most ${n \choose 2}$ rotations in $\Pi$.
  \end{enumerate}
\end{claim}
\begin{proof}
  %(\ref{subclaimEachCoverHasRot}) The fact that a rotation exists follows
  %directly from claims~\ref{claimCoveringCondition}
  %and~\ref{claimRotationFromMPDA}. Furthermore, it's easy to see that
  %the pair of matchings uniquely determines the rotation.

  (\ref{subclaimStabPartners}) 
  The ``if'' direction is true by definition.
  For the ``only if'' part,
  let $\mu_0$ be a matching other than the
  woman-optimal outcome, and let $\mu_0(m)=w$ for $(m,w)$ not paired in the
  woman-optimal outcome.
  Let $\mu'$ be the woman-optimal stable outcome.
  % any stable matching in which $w$ is paired to her next-best
  % stable partner after $m$.
  Consider any maximal chain $\mu_0 < \mu_1 < \ldots < \mu_k = \mu'$ between
  $\mu_0$ and $\mu_k$ (i.e. $\mu_i$ covers $\mu_{i-1}$ for each $i$).
  Because $w$ is not matched to $m$ in $\mu'$, there must be some covering
  relation $\mu_{i-1} < \mu_i$ where $w$ is at $m$ in $\mu_{i-1}$ but not in
  $\mu_i$. By claim~\ref{claimRotationFromMPDA}, this corresponds to a rotation
  in which $(m,w)$ appears.

  % Consider $\mu' = \bigwedge \{ \mu' : \mu' \ge \mu_0, \mu'(w)\succ_w \mu_0(w) \}$,
  % i.e. the meet of all stable matchings $\ge\mu_0$ where $w$ does strictly better.
  % This is well-defined because we've assumed that set is nonempty.
  % We claim that $\mu'$ covers $\mu_0$. Proof: If $\mu_0 < \mu < \mu'$,
  % then

  % Then there exists a $\mu'$ which covers $\mu_0$, and in which $w$ receives
  % a different match than $m$ by claim~\ref{claimCanGoUpIfNotWosm}.
  % The corresponding rotation includes $(w,m)$.
  % The converse is true by definition.

  (\ref{subclaimNextStablePartner}) Let the rotation $\rho$
  be exposed in $\mu_0$ and let the elimination of $\rho$ from $\mu_0$
  be $\mu'$.
  Suppose for the sake of contradiction
  that $w_i$ has a stable partner $m^*$ who she ranks between
  $m_{i-1}$ and $m_i$, i.e. $m_{i}\prec_w m^*\prec_w m_{i-1}$.
  Let $\mu^*$ pair $w_i$ and $m^*$.
  % We use the lattice-theoretic properties of $\L$,
  % as overviewed in~\ref{theoremDistributiveLattice}.
  Consider the matching $\mu = (\mu_0 \vee \mu^*)\wedge \mu'$.
  We have $\mu \le \mu'$, and because $\mu_0 \le \mu'$
  and $\mu_0 \le \mu_0 \vee \mu^*$, we also get $\mu_0 \le \mu$.
  Because $\mu(w_i) = m^*$, that means
  $\mu_0 < \mu < \mu'$, which contradicts the fact that $\mu'$ covers $\mu_0$.
  This proves that $m_{i-1}$ is the worse stable partner of $w_i$ after $m_i$.
  The proof that, for each man $m_i$, $w_{i+1}$ is the next best stable partner
  $m_i$ has below $w_i$ is analagous.

  (\ref{subclaimAtMostOnceRotate}) By part~\ref{subclaimNextStablePartner},
  given one pair $(m_i,w_i)$, the value of $w_{i+1}$
  is uniquely determined as the next stable partner of $m_i$ below $w_i$.
  But then, by considering $m_i$ and $w_{i+1}$, the value of
  $m_{i+1}$ is uniquely determined as the worst stable partner 
  of $w_{i+1}$ before $m_i$. Continuing this process, we see that
  specifying one pair in $\rho$ uniquely determines all of $\rho$.

  (\ref{subclaimNumberRotations}). By part~\ref{subclaimStabPartners},
  each man can appear with at most
  $n-1$ agents in some rotation (each of his stable partners except his partner
  in the woman-optimal outcome), and by part~\ref{subclaimAtMostOnceRotate}
  that pair can appear at most once.
  At least two pairs of agents appear in
  each rotation, so the total number of rotations is at most $n(n-1)/2$.

\end{proof}

We already know a decent amount about the structure of individual rotations.
However, the rotations interact in a specified way.
In particular, there is a natural \emph{ordering} among them --
some rotations must be eliminated before others.
To make this precise, we need some definitions.

\begin{definition} \label{def:typedEdges}
  Let $\rho_1$ and $\rho_2$ be rotations in $\Pi$.
  \begin{enumerate}
    \item If there exists a man/woman pair $(w,m)$ such that
      $\rho_1$ moves $m$ to $w$, and $(w,m)$ appear in
      $\rho_2$ (i.e. $\rho_2$ moves $m$ away from $w$),
      then $\rho_1$ is called a \emph{type 1 predecessor}
      of $\rho_2$
    \item If there exist a man/woman pair $(w,m)$ such that:
      \begin{itemize}
        \item $\rho_1$ moves $w$ from $m_j$ to $m_{j-1}$,
          and $m_{j} \prec_w m \prec_w m_{j-1}$ \\
          ($\rho_1$ moves $w$ from below to above $m$)
        \item $\rho_2$ moves $m$ from $w_i$ to $w_{i+1}$,
          and $w_i \succ_m w \succ_m w_{i+1}$ \\
          ($\rho_2$ moves $m$ from above to below $w$)
      \end{itemize}
      Then $\rho_1$ is called a \emph{type 2 predecessor} of $\rho_2$.
  \end{enumerate}

  If $\rho_1$ is either a type 1 or type 2 predecessor of $\rho_2$,
  we say $\rho_1$ is a predecessor of $\rho_2$.

  Define the predecessor graph $G(\Pi)$ on rotations as follows:
  the vertices are every rotation $\rho\in\Pi$,
  for every $\rho_1, \rho_2$ such that $\rho_1$ is a predecessor of $\rho_2$,
  there is a directed edge from $\rho_1$ to $\rho_2$ 
  (labeled according to whether they are predecessors of type 1 or type 2 (or both)).
\end{definition}

In short, $\rho_1$ is a type 1 predecessor of $\rho_2$ if $\rho_1$ move a
couple $(m,w)$ together who $\rho_2$ moves apart.
In this case, $\rho_1$ must be eliminated first by definition.
% (so $\rho_1$ must come first so that those two agents  .
Intuitively, $\rho_1$ is a type 2 predecessor of $\rho_2$ 
if running $MPDA$ to eliminate $\rho_2$ would trigger the elimination
of $\rho_1$. Otherwise, for the pair $(m,w)$ in the definition,
$m$ would propose to $w$ as he moves from $w_{i}$ to $w_{i+1}$,
and $w$ would accept that proposal and
trigger the elimination of $\rho_1$ 
(finally giving $w$ an even better match than $m$).

It turns out that the above two types of predecessor relations 
are necessary and sufficient
to characterize which rotations must be eliminated before each other.
More precisely, a permutation of the set of all rotations
can be eliminated, one after the other,
if and only if they are topologically sorted in the graph $G(\Pi)$\footnote{
  A topological sort of a directed acyclic graph
  is a permutation of the vertices of the graph such that, for each directed
  edge $(u,v)$ in the graph, $u$ comes before $v$.
}.

We first prove that every possible sequence of eliminations forms
a topological sort of $G(\Pi)$. As we will formally spell out
in~\ref{theoremRotationsBijection}, this means that topological sorts
of $G(\Pi)$ suffice to represent all stable matchings.

\begin{claim}\label{claimChainsAreToplSorts}
  % Let $\mu_0$ be the man-optimal stable matching and
  % let $\mu$ be any stable matching.
  Consider any chain $\mu_0 < \mu_1 < \ldots < \mu_k$ in $\L$
  where $\mu_0$ is man-optimal, $\mu_k$ is woman-optimal,
  and $\mu_{i+1}$ covers $\mu_i$ for each $i$
  (i.e. consider a maximal chain in $\L$).
  Then
  \begin{enumerate}
    \item \label{subclaimConvertToChains}
      Then there exists a unique sequence of rotations
      $\rho_0, \rho_1, \ldots, \rho_{k-1} \in \Pi$ such that
      $\mu_{i+1}$ is the elimination of $\rho_i$ from $\mu_i$ for each $i$.
    \item \label{subclaimAllRotationsAppear}
      Every rotation in $\Pi$ appears exactly once in this sequence.
    \item \label{subclaimPredecesorsBefore}
      If $\rho$ is a predecessor of $\rho^*$ (of type 1 or type 2),
      then $\rho$ appears before $\rho^*$ in this sequence.
  \end{enumerate}
\end{claim}
\begin{proof}
  (\ref{subclaimConvertToChains})
  By claim~\ref{claimRotationFromMPDA}, such a $\rho_i$ exists for each $i$.
  Furthermore, given $\mu_i$ and $\mu_{i+1}$, it's clear that $\rho_i$
  is uniquely determined.

  (\ref{subclaimAllRotationsAppear})
  We showed in claim~\ref{claimStablePartnerRotations},
  part~\ref{subclaimNextStablePartner} that rotations
  (and hence covering relations) move agents up or down
  one place on their list of stable partners.
  Thus, over the course of the maximal chain, every stable pair
  % who are not matched in the woman-optimal outcome
  must be matched in some $\mu_i$
  (or else those agents could not reach their match in the woman-optimal
  outcome).
  % , and then move (unless
  % they are paired in the woman-optimal).
  Moreover, each stable pair which is not matched in $\mu_{k}$
  must appear in some rotation $\rho_i$.
  Because a stable pair appears in at
  most one rotation (claim~\ref{claimStablePartnerRotations},
  part~\ref{subclaimAtMostOnceRotate}), this means every rotation in $\Pi$ is in
  this sequence.

  (\ref{subclaimPredecesorsBefore})
  Assume for contradiction that $\rho_j$ is a predecessor of $\rho_i$ for $i<j$.
  % but $\rho$ appeared after $\rho^*$ in the list.
  We have two cases.

  Suppose $\rho_j$ is a type 1 predecessor of $\rho_i$.
  By definition, there exists a pair $(m,w)$ such that
  $\rho_j$ moves $m$ to $w$, and $(m,w)$ appears in $\rho_i$.
  By~\ref{claimStablePartnerRotations},
  part~\ref{subclaimNextStablePartner},
  rotations always move women to men which they rank higher than
  their current match, this means that in $\mu_j$, $w$ was matched below $m$
  (i.e. $m \succ_w \mu_j(w)$). But $(m,w)$ are matched in $\mu_i$.
  Thus, we cannot have $\mu_i \le \mu_j$, a contradiction.

  Now suppose $\rho_j$ is a type 2 predecessor of $\rho_i$.
  By definition, there exist $(m,w)$ such that
  $\rho_j$ moves $w$ from below $m$ to above $m$  
  and $\rho_i$ moves $m$ from above $w$ to below $w$.
  Thus, in $\mu_j$, $w$ is matched below $m$,
  and in $\mu_{i+1}$, $m$ is matched below $w$.
  Now, $\mu_{i+1}\le \mu_j$, so in $\mu_j$, $w$ is also matched below $m$.
  But this means that $\mu_{i+1}$ is not stable,
  a contradiction.
\end{proof}

Note that part~\ref{subclaimPredecesorsBefore} above implies that
$G(\Pi)$ is an acyclic graph.

% Because the elimination of a rotation is easily computable, the previous claim
% shows us that $\Pi$ compactly represents all stable matchings.
We'll see next that \emph{only} stable matchings arise in the way described by
the previous claim. %, and that we can efficiently find $G(\Pi)$.
In other words, the type 1 and type 2 predecessor relations are the \emph{only}
issues to applying any sequence of rotations you would like.

Our strategy will be to show that every topological sort of $G(\Pi)$
corresponds to a maximal chain in $\L$.
We start with an arbitrary chain which corresponds
to some topological sort, and the apply a special type of
``commutativity operation'' in order to transform that initial
topological sort into the one we want (while preserving the property of
corresponding to some maximal chain along the way).
The next claim is a technical lemma which provides the type of 
commutativity operation needed.
Another way to summarize this claim is that two adjacent rotations 
which are not predecessors do not interfere with each other. 
% \\ \vspace{-0.3in}
% \InsertBoxR{2}{
% }[5]
  % \begin{wrapfigure}{rh}{3cm}
  % \end{wrapfigure}

\hspace{-0.25in}
\begin{minipage}{0.8\textwidth}
\begin{claim} \label{claimToplSortsInterchange}
  Suppose $\mu_0 < \mu_1 < \mu_2$, where each matching covers the previous one.
  Let $\rho_0$ and $\rho_1$ be the corresponding rotations,
  i.e. $\mu_{i+1}$ is the elimination of $\rho_i$ from $\mu_i$ for $i=0,1$.
  Assume that $\rho_0$ is not a predecessor
  of $\rho_1$ in $G(\Pi)$.
  Then $\rho_1$ is also exposed in $\mu_0$.
  Moreover, if $\mu_1'$ is the elimination of $\rho_1$ from $\mu_0$,
  then $\rho_0$ is exposed in $\mu_1'$, and $\mu_2$ is the elimination of
  $\rho_0$ from $\mu_1'$.
\end{claim}
\end{minipage}
\quad
\begin{minipage}{0.2\textwidth}
\begin{tikzpicture}[node distance=2cm, semithick, auto] 
  \node (u1) at (0,0) {$\mu_1$};
  \node (u2) at (1,1.5) {$\mu_2$};
  \node (u0) at (1,-1.5) {$\mu_0$};
  \node (u1') at (2,0) {$\mu_1'$};
  \path[->] (u1) edge[above left] node {$\rho_1$} (u2)
            (u1') edge[above right] node {$\rho_0$} (u2)
            (u0) edge[below right] node {$\rho_1$} (u1')
            (u0) edge[below left] node {$\rho_0$} (u1); 
\end{tikzpicture}
\end{minipage}

\begin{proof}
  Recall that every agent who does not appear in $\rho_0$ receives the same
  match in $\mu_0$ and $\mu_1$.
  % From the previous claim we know that $\rho_1$ is not a predecessor
  % of $\rho_0$, so neither rotation is a type 1 predecessor of the other.

  First, we claim that no agent can appear in both rotations.
  Proof: If some man $m$ appeared in both rotation,
  then $\rho_0$ moved $m$ to some woman $\mu_1(m)$,
  and then $(\mu_1(m), m)$ must appear in $\rho_1$.
  % so $\rho_0$ is a type 1 predecessor of $\rho_1$.
  If a woman $w$ appears in both rotations,
  then the pair $(w, \mu_1(w))$ appears in $\rho_1$,
  so $\rho_0$ must move $\mu_1(w)$ to $w$.
  Because we've assumed that $\rho_0$ is not a type 1 predecessor of $\rho_1$,
  neither of the above cases can occur.

  Now, let $w_0$ be a woman appearing in $\rho_1$ 
  and consider $MPDA(P_{w_0}(\mu_0))$.
  We claim that this produces a stable outcome
  and no woman receives more than one proposal from a man she prefers to her
  match in $\mu_0$.
  Proof: Let $\rho_1 = [(w_0,m_0), (w_1,m_1), \ldots, (w_{k-1},m_{k-1})]$.
  In $MPDA(P_{w_0}(\mu_0))$, the free man is initially $m_0$.
  Consider the proposals that $m_i$ makes after he is rejected by $w_i$.
  Some of the women he proposes to may be at different matches in $\mu_0$ than
  in $\mu_1$ (specifically, those women who were moved by $\rho_1$).
  However, for all such women $w$ who $m_i$ ranks above $w_{i+1}$,
  $w$ cannot be matched below $m_i$,
  or else $\rho_0$ would be a type 2 predecessor of $\rho_1$ by definition.
  Because none of the agents in $\rho_2$ appear in $\rho_1$,
  $w_{i+1}$ has the same match in $\mu_0$ as in $\mu_1$.
  Thus, all the women $m_i$ proposes to before $w_{i+1}$ will reject him,
  but $w_{i+1}$ will accept him.
  Thus, by induction, the rejection chain of $MPDA(P_{w_0}(\mu_0))$
  will be exactly the same as in $MPDA(P_{w_0}(\mu_1))$.
  By claim~\ref{claimRotationFromMPDA}, this means $\rho_1$ is exposed in
  $\mu_0$ and that
  and $\mu_1' = MPDA(P_{w_0}(\mu_0))$ is the
  elimination of $\rho_1$ in $\mu_0$.

  Finally, consider running $MPDA(P_{w_0'}(\mu_1'))$ for some $w_0'$ appearing in
  $\rho_0$. Let $\rho_1 = [(w_0',m_0'),\\ \ldots, (w_{k'-1}',m_{k'-1}')]$,
  and again consider a free man $m_i'$ during this rejection chain.
  The only difference between $\mu_1'$ and $\mu_0$ is that the women who appear
  in $\rho_1$ have received better partners. However, $w_{i+1}$ is still matched
  to $\mu_0(w_{i+1})$, again because the agents in $\rho_0$ and $\rho_1$ 
  are disjoint. The women $m_i$ proposes to before $w_{i+1}$ can only have
  higher matches than in $\mu_0$, so they will still reject his proposals.
  But $w_{i+1}$ will still accept.
  % Tus, during this rejection chain no man will be
  % accepted by any woman other than the woman he moves to in $\rho_0$.
  Thus, $MPDA(P_{w^*}(\mu_1'))$ will terminate with exactly $\rho_1$ eliminated from
  $\mu_1'$, and no woman will recieve multiple proposals from a man
  she prefers to her match in $\mu_1'$.
  So $\rho_1$ was exposed in $\mu_1'$.
  % Because $\rho_1$ is not a type 2 predecessor of $\rho_0$,
  % there exists no man/woman pair $(m,w)$ where $\rho_1$ moves $w$ above $m$
  % and $\rho_0$ moves $m$ below $w$.

  Finally, it's clear from the definitions that the elimination of $\rho_0$
  from $\mu_1'$ is $\mu_2$ (the match of every agent is uniquely determined as
  either the match from $\mu_0$ or the match which is uniquely 
  specified in $\rho_0$ or $\rho_1$).
\end{proof}

Now we can prove that only maximal chains arise from topological sorts of
$G(\Pi)$. One short way to summarize this proof is the following:
we can transform any two topological sorts of $G(\Pi)$ between each other
using only the ``adjacent order swapping'' operation given by the previous
lemma\footnote{
  One way to do this (different then outlined in our formal proof)
  is to label the first list of rotations with $0, 1, \ldots, k-1$,
  then simply \emph{bubble sort} the second list of rotations.
}. Thus, starting from a fixed topological sort of $G(\Pi)$
(which corresponds to a maximal chain by claim~\ref{claimChainsAreToplSorts})
we see than any other topological sort will also correspond to a maximal chain.

\begin{claim} \label{claimToplSortsAreChains}
  Consider any topological sort $\rho_0, \rho_1, \ldots, \rho_{k-1}$ of
  $G(\Pi)$, i.e. an ordering of each element of $\Pi$ such that whenever
  $\rho_i$ is a predecessor of $\rho_j$, we have $i < j$.
  Then this sequence corresponds to a maximal chain $\mu_0, \mu_1, \ldots, \mu_{k}$
  in the stable matching lattice $\L$ such that
  $\mu_0$ is the man-optimal stable outcome,
  $\mu_{i+1}$ is the elimination of $\rho_i$ from $\mu_i$,
  and $\mu_k$ is the woman-optimal outcome.
\end{claim}
\begin{proof}
  For this proof, say that a permutation $\rho_0', \ldots, \rho_{k-1}'$
  of the set $\Pi$ is \emph{valid} if there exists 
  $\mu_0', \mu_1', \ldots, \mu_{k}'$ a maximal chain in $\L$
  such that $\mu_{i+1}$ is the elimination of $\rho_i$ from $\mu_i$ for each $i$.

  Fix any arbitrary maximal chain $\mu_0 < \mu_1 < \ldots < \mu_k$ in $\L$.
  By claim~\ref{claimChainsAreToplSorts}, there exists a corresponding
  valid sequence $\rho_0, \rho_1, \ldots, \rho_{k-1}$ 
  of rotations which is a topological sort of $G(\Pi)$.
  Now, given any permutation of $0, 1, \ldots, k-1$, say
  say ${i_0}, {i_1}, \ldots, {i_{k-1}}$,
  we'll prove that $\rho_{i_0}, \rho_{i_1}, \ldots, \rho_{i_{k-1}}$ is valid
  using induction on the \emph{number of inversions} of the permutation
  (i.e. the number of pairs $j < k$ such that $i_j > i_k$).
  If there are no inversions, then $i_j = j$ for each $j$ and we are done.

  Now, suppose $I = {i_0}, {i_1}, \ldots, {i_{k-1}}$ has at least one inversion
  and let $i_{j}, i_{j+1}$ be any \emph{adjacent} inverted pair
  (if no adjacent inverted pairs exist, then no inverted pairs can exists).
  Consider the ordering 
  $I' = i_0, \ldots,\allowbreak i_{j-1}, i_{j+1}, i_j, i_{j+2}, 
  \allowbreak \ldots, {i_{k-1}}$.
  There is exactly one fewer inverted pair in this new
  ordering than in the original one, so by induction the 
  ordering on rotations corresponding to $I'$ is valid.
  Say this ordering on rotations corresponds to a maximal chain containing
  the matchings $\mu_a < \mu_b < \mu_c$, where $\mu_b$ is the elimination
  of $\rho_{i_{j+1}}$ from $\mu_a$ and $\mu_c$ is the elimination of 
  $\rho_{i_j}$ from $\mu_b$.
  Because both $\rho_0, \ldots, \rho_{k-1}$
  and $\rho_{i_0}, \ldots, \rho_{i_{k-1}}$ are topological sorts of $G(\Pi)$,
  $\rho_{i_j}$ and $\rho_{i_{j+1}}$ cannot be in a predecessor relation.
  Thus, applying claim~\ref{claimToplSortsInterchange} to the covering
  relations $\mu_{a} < \mu_{b} < \mu_{c}$,
  we see that the original ordering
  $\rho_{i_0}, \rho_{i_1}, \ldots, \rho_{i_{k-1}}$ is also valid
  (and only the matching $\mu_b$ is different along the 
  corresponding maximal chain).

  Thus, every topological sort of $G(\Pi)$ is valid by induction.

  % Now, iteratively apply
  % to change the sequence $\rho'_0, \rho'_1, \ldots, \rho'_{k-1}$
  % to $\rho_0, \rho_1, \ldots, \rho_{k-1}$.
  % At every step along the way, we swap two adjacent rotations
  % $\rho_i$ and $\rho_{i+1}$ in the list, and we change one of the
  % matchings, specifically $\mu_i$.
  % However, we preserve the property that each matching covers the previous one,
  % and thus the new ordering of rotations is still valid.
  % Once we reach $\rho_0, \rho_1, \ldots, \rho_{k-1}$,
  % we've proven the claim.

  % Why is it always possible to change one sequence into another?
  % Here's one way:
  % Consider the rotation $\rho_0$, which be a source node in $G(\Pi)$
  % as it comes first in a topological sort.
  % Because $\rho_0$ has no predecessors, it can be moved directly to the front of 
  % the list starting from $\rho'_0, \rho'_1, \ldots, \rho'_{k-1}$.
  % Now, ignoring $\rho_0$, $\rho_1$ cannot have any predecessors in $G(\Pi)$,
  % so it can be moved to the second place.
  % Continuing in this manner will succesfully transition from
  % $\rho'_0, \rho'_1, \ldots, \rho'_{k-1}$ to $\rho_0, \rho_1, \ldots, \rho_{k-1}$.
  % Put another way, just run an insertion sort on
  % $\rho'_0, \rho'_1, \ldots, \rho'_{k-1}$.
\end{proof}

Finally, after one more simple definition,
we arrive at our long sought after bijection.

\begin{definition}
  A \emph{closed} subset $S$ of $G(\Pi)$ is a collection of rotations
  such that, whenever $\rho_2$ is in $S$ and $\rho_1$ is a predecessor of
  $\rho_2$, then $\rho_1$ is in $S$.
\end{definition}

\begin{theorem} \label{theoremRotationsBijection}
  % Let the set of rotations $\Pi$ be partially ordered by the transitive closure
  % of the type 1 and type 2 predecessor relations in $G(\Pi)$.
  There is a bijection between the collection of closed
  subsets of $G(\Pi)$ and the stable matching lattice $\L$.

  This bijection is given as follows:
  for a closed subset $S$ of $\Pi$, let
  $\rho_0,\ldots,\rho_i$ be a topological sort of $S$ in $G(\Pi)$.
  Let $\mu_0$ be the man-optimal stable outcome,
  and let $\mu_{j+1}$ be the elimination of $\rho_j$ from $\mu_j$ for each
  $j=0,\ldots,i$.
  Then the matching corresponding to $S$ is given by $\mu_{i+1}$.

  Furthermore, let $\mu_1$ and $\mu_2$ be stable matchings corresponding
  to $S_1$ and $S_2$ respectively. Then $\mu_2$ woman-dominates $\mu_1$
  (i.e. $\mu_2 \ge \mu_1$)
  if and only if $S_2 \supseteq S_1$.
\end{theorem}
\begin{proof}
  First, we show that this correspondence is surjective.
  Given a matching $\mu$, consider a maximal chain
  $\mu_0 < \mu_1 < \ldots < \mu_k$ containing it, say $\mu = \mu_i$.
  By claim~\ref{claimChainsAreToplSorts}, there exists a corresponding
  sequence of rotations $\rho_0, \rho_1, \ldots, \rho_{k-1}$.
  Consider the set $S = \{\rho_j\}_{j < i}$.
  For each $\rho_j\in S$ and $\rho'$ a predecessor of $\rho$,
  $\rho'$ must also be in $S$ because $\rho_0, \rho_1, \ldots, \rho_{k-1}$
  is a topological sort.
  So $S$ is closed in $G(\Pi)$.
  Furthermore, $\mu = \mu_i$ is exactly given by the successive elimination of
  the rotations $\rho_0, \ldots, \rho_{i-1}$, starting from $\mu_0$.
  So $S$ corresponds to $\mu$.

  Next, we show the correspondence is injective.
  Let $S_1,S_2$ be distinct closed subsets of $G(\Pi)$,
  let matching $\mu_1, \mu_2$ correspond to $S_1, S_2$,
  and without loss of generality take $\rho \in S_1\setminus S_2$.
  Because $\rho$ has been eliminated in $\mu_1$ but not in $\mu_2$,
  any woman $w$ appearing in $\rho$ must prefer their match in $\mu_1$
  to their match in $\mu_2$.
  Thus, $\mu_1$ and $\mu_2$ cannot be the same matching.

  Finally, let $\mu_1, \mu_2\in \L$ correspond to $S_1, S_2 \subseteq G(\Pi)$.
  We have $\mu_1\le\mu_2$ if and only if each woman $w$ 
  does at least as well in $\mu_2$
  as in $\mu_1$. For a fixed woman $w$, this occurs
  if and only if every rotation involving $w$ which appears in $S_1$
  also appears in $S_2$.
  This is equivalent to the condition
  that every rotation which appears in $S_1$ also appearing in $S_2$,
  i.e. $S_1\subseteq S_2$.

  % On the other hand, consider a closed set $S$ of $G(\Pi)$.
  % Consider a topological sort of $S$,
  % say $\rho_0,\ldots, \rho_{i}$,
  % and extend it to a topological sort of the whole $G(\Pi)$,
  % say $\rho_0,\ldots, \rho_{i},\ldots, \rho_{k-1}$.
  % By claim~\ref{claimToplSortsAreChains}, this corresponds to a maximal chain
  % $\mu_0,\mu_1,\ldots,\mu_k$,
  % and $\mu_{i+1}$ is a stable matching which
  % is the elimination of the rotations $\rho_0,\ldots,\rho_{i-1}$.

  % Thus, there is a bijection between the closed subsets of $G$ and the stable
  % matching lattice $\L$.
\end{proof}

\paragraph{Remark:} Because the bijection above respects ordering
(i.e. $\mu_2 \ge \mu_1$ if and only if $S_2 \supseteq S_1$),
the bijection above is actually a \emph{lattice isomorphism}.
So joins and meets in $\L$ correspond to joins and meets in the lattice of 
close subsets of the graph $G(\Pi)$, which is given by
set union and set intersection, respectively.

% \paragraph{Remark:} with some extra work, you can show that this bijection
% even preserves joins and meets, in the sense that the intersection of two sets
% $S$ and $T$ corresponds to the meet of the corresponding matchings.

%% file: AlgAndComments.tex
\subsection{A simple example}

% example graphs
  % how to change width and height 
  % \let\oldemph\emph
  % \renewcommand{\emph}[1]{\textbf{\oldemph{#1}}}

\begin{figure}[htp!]
 \begin{minipage}{1.0\textwidth}
  \begin{tabular}[c]{m{6cm} | m{5cm} | m{4cm} }
    \begin{tabular}[h]{c}
    \\
      \begin{tikzpicture}[node distance=2.5cm, semithick, auto, font=\small] 
        \node[state, fill=white] (w1) {$w_1$};
        \node[state, fill=white] (w2) [right of=w1] {$w_2$};
        \path[->] (w1) edge[bend left, above] node {$m_1$} (w2)
            (w2) edge[bend left, below] node {$m_2$} (w1);
        \coordinate (w1') at ($(w1.north) + (-1.4em, +0.8em)$);
        \coordinate (w2') at ($(w2.south) + (+1.4em, -0.8em)$);
        \node [ultra thick, dashed, draw=red, fit= (w1') (w2'), fill=red!20, fill opacity=0.2] (r1) {};   
        \node [yshift=1.5ex, xshift=3.0em, red, font=\small] at (r1.south) {$\rho_1$};
      \end{tikzpicture}
      \\ \\ \\ \\
      \begin{tikzpicture}[node distance=1.75cm, semithick, auto, font=\small, cross/.style={cross out, draw, 
         minimum size=4*(#1-\pgflinewidth), 
         inner sep=0pt, outer sep=0pt}, cross/.default={3.5pt}]
        % If anyone has to edit this, I'm sincerely sorry for the names.
        \node[state, fill=white] (w1) {$w_1$};
        \node[state, fill=white] (w3) [right of=w1] {$w_3$};
        \node[state, fill=white] (w5) [above right of=w3] {$w_4$};
        \node[state, fill=white] (w4) [below right of=w5]  {$w_5$};
        \node[state, fill=white] (w5c) at ($(w1) + (0, -2.4)$) {$w_6$};
        \coordinate (w4') at ($(w4.south) + (+1.4em, -0.8em)$);
        \coordinate (w1') at ($(w1.north) + (-2.0em, 0.9em)$);
        \coordinate (w5c') at ($(w5c.north) + (2.0em, -2.9em)$);

        \node[ultra thick, cross, red] (x) at ($(w1)!0.5!(w3) + (0, 0.4)$) {};
        \node [ultra thick, dashed, draw=gray!70, fit= (w3) (w5) (w4'), fill=gray!20, fill opacity=0.2] (r2) {};
        \node [ultra thick, dashed, draw=red, fit= (w1') (w5c'), fill=red!20, fill opacity=0.2] (r3) {};
        \path[->] (w1) edge[bend left, above left] node {$m_2$} (w3)
                (w3) edge[bend left, above] node {$m_3$} (w5)
                (w5) edge[bend left, below] node {$m_4$} (w4)
                (w4) edge[bend left, below] node {$m_5$} (w3)
                
                (w1) edge[bend left, above] node {$m_2$} (w5c)
                (w5c) edge[bend left, below] node {$m_6$} (w1);

        \node [yshift=-2.0ex, xshift=-3.0em, gray!70, font=\small] at (r2.north) (r2label) {$\rho_2$};
        \node [yshift=2.0ex, xshift=-1.4em, red, font=\small] at (r3.south) (r3label) {$\rho_3$};
      \end{tikzpicture}
    \end{tabular}
    & 
    \begin{tabular}[h]{c}
     
      \begin{tikzpicture}[node distance=1.75cm, semithick, auto, font=\small] 
        % If anyone has to edit this, I'm sincerely sorry for the names.
        \node[state, fill=white] (w2) {$w_2$};
        \node[state, fill=white] (w3) [right of=w2] {$w_4$};
        \node[state, fill=white] (w5) [right of=w3] {$w_5$};
        \node[state, fill=white] (w4) at ($(w3)!0.5!(w5) + (0, -1.4)$) {$w_3$};
        \coordinate (w3') at ($(w3.north) + (-1.2em, 0.9em)$);
        \path[->] (w2) edge[bend left, above] node {$m_1$} (w3)
                  (w3) edge[bend left, above] node {$m_4$} (w5)
                  (w5) edge[bend left, below] node {$m_5$} (w4)
                  (w4) edge[bend left, below] node {$m_3$} (w3)
                  (w3) edge[bend left, above] node {$m_1$} (2.5,1.5);
        \node [ultra thick, dashed, draw=red, fit= (w3') (w5) (w4), fill=red!20, fill opacity=0.2] (r2) {};
        \node [yshift=2.0ex, xshift=+3.0em, red, font=\small] at (r2.south) {$\rho_2$};
      \end{tikzpicture}
      \\ \\ \\ \\ 
      \begin{tikzpicture}[node distance=2.5cm, semithick, auto, font=\small, rect node/.style={draw, minimum size=1cm}, scale = 1] \label{dependency} 
        \node[scale=1.3] (r1) {\large $\rho_1$};
        \node[scale=1.3] (r2) at (2,0) {\large $\rho_2$};
        \node[scale=1.5] (r3)  at (1, 2) {$\rho_3$};
        \path[->] (r3) edge[above left] node {Type 1} (r1)
                  (r3) edge[above right] node {Type 2} (r2);
      \end{tikzpicture} 
      \\ \\ \\
    \end{tabular}
    &
    \begin{tabular}[c]{l}
      \begin{tabular}{c | c c  c c }
        $m_1$ & $w_1$ & \multicolumn{1}{c|}{$w_2$} & $w_4$ \\
        $m_2$ & $w_2$ & \multicolumn{1}{c|}{$w_1$} & $\mathbf{w_3}$ & $w_6$\\
        \cline{2-3}
        $m_3$ & $w_3$ & \multicolumn{1}{c|}{$w_4$} & &\\
        $m_4$ & $w_4$ & \multicolumn{1}{c|}{$w_5$} & &\\
        $m_5$ & $w_5$ & \multicolumn{1}{c|}{$w_3$} & &\\
        \cline{2-3}
        $m_6$ & $w_6$ & $w_1$
      \end{tabular}
      \\ \\
      \begin{tabular}{c | c c c c}
        $w_1$ & \multicolumn{1}{c|}{$m_6$} & $m_2$ & $m_1$ \\
        \cline{2-2}
        $w_2$ & $m_1$ & $m_2$ \\
        \cline{2-4}
        $w_3$  & $m_5$ & {$\mathbf{m_2}$} & $m_3$ \\
        $w_4$ & {$m_3$} & $m_1$ & $m_4$ \\
        $w_5$ & {$m_4$} & $m_5$ \\
        \cline{2-3}
        $w_6$ & $m_2$ & $m_6$
      \end{tabular}
      \\ \\
      \begin{tikzpicture}[scale=1]
        \node (zero) at (0,-1) {$(123456)$};
        \node (l) at (-1,0) {$({\bf21}3456)$};
        \node (r) at (1,0) {$(12{\bf 453}6)$};
        \node (b) at (0,1) {$(214536)$};
        \node (one) at (0,2) {$(2{\bf6}453\bf{1})$};
        \draw (zero) -- (l) -- (b) -- (one);
        \draw (zero) -- (r) -- (b);
      \end{tikzpicture}

    \end{tabular}
  \end{tabular}
  
  \caption{
    A sample instance with rotation poset and stable matching lattice.
    % The rotations are highlighted through borders in the preference lists,
    % and the type 2 predecessor relation is caused by the boldfaced entries.
  }
  \label{fig:posetExample}
  \end{minipage}
\end{figure}
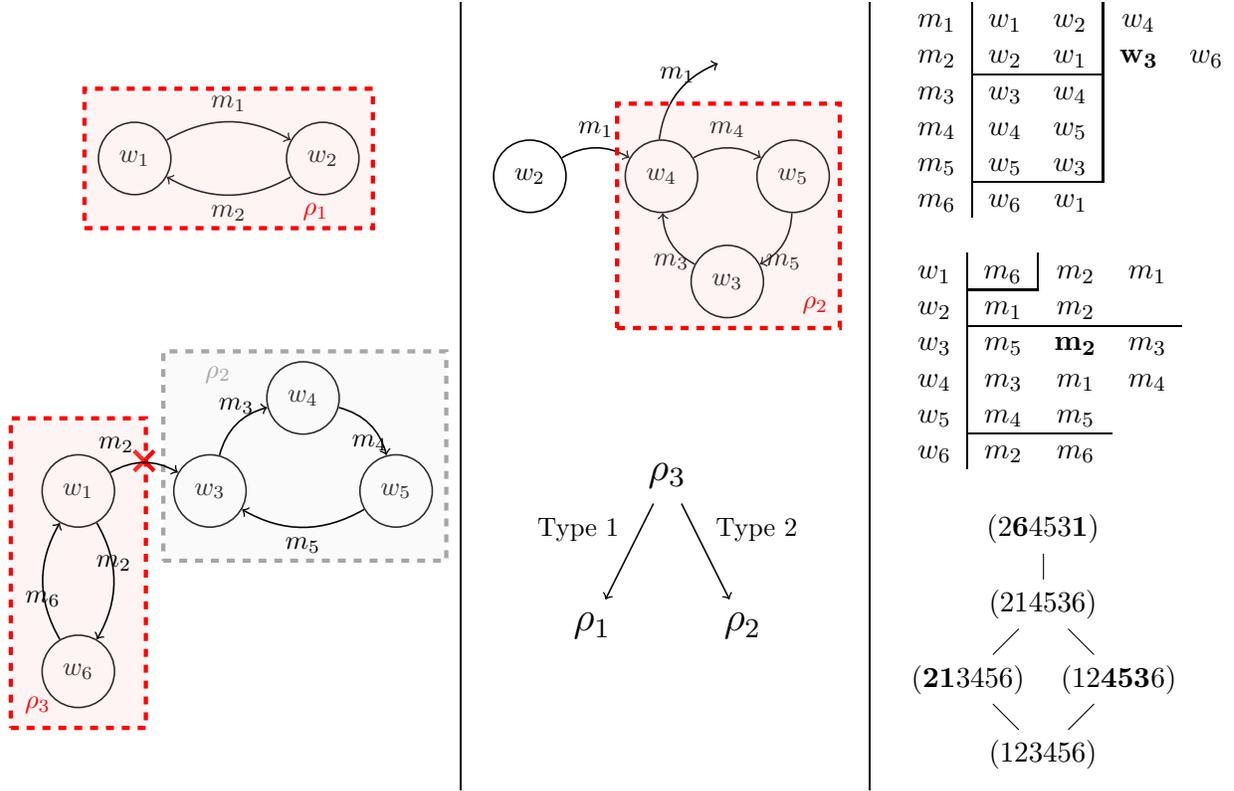

 We start with an example of how to use the facts proven above. Let men and
 women's preference list be as illustrated in the right column of
 Figure~\ref{fig:posetExample}. 
 The borders in the table highlight the preferences which cause the different
 rotations to form, and the boldfaced entries correspond to a type 2 predecessor
 relationship.
 
 % \textbf{Identifying rotations} 
 For the sake of illustration, let us denote a
 stable matching $\mu$ as $(\mu(m_1) \mu(m_2) ... \mu(m_n))$ (e.g. $(123456)$
 means every $m_i$ is matched to $w_i$). 
 It is easy to see that $\mu_0 = (123456)$ 
 is the man-optimal stable matching in this example. 

 Now, imagine running $MPDA(P_{w_1}(\mu_0))$ (where woman $1$ truncate her list
 just above $m_1$). 
 The rejection chain is $(w_1, m_1, w_2, m_2, w_1)$,
 and each woman recieves at most one proposals from a man she prefers to her
 match in $\mu_0$, so by claim~\ref{claimRotationFromMPDA} we've
 discovered the rotation
 $\rho_1 = [(w_1, m_1) , (w_2, m_2)]$.
 % The $MPDA$ algorithm starts from the tentative matching
 % $(123456)$. Next $w_1$ rejects its original stable partner $m_1$, who then
 % propose to $w_2$. $w_2$ accepts the proposal and rejects $m_2$. $m_2$ now
 % propose to its next desired partner $w_1$, who accepts the proposal. Now the
 % algorithm terminates, and we have found a new stable matching $\mu_1 =
 % (213456)$. Since in $MPDA(P_{w_1}(\mu_0))$, no women receives more than one
 % proposal from men who she strictly prefers to her match in $\mu_0$. By
 % Claim~\ref{claimCoveringCondition}, $\mu_1$ covers $\mu_0$, and thus we have
 % discovered a rotation cause the change from $\mu_0$ to $\mu_1$: 
 
 Next we run $MPDA(P_{w_2}(\mu_1))$.
 The rejection chain is $(w_2, m_1, w_4, m_4, w_5, m_5, w_3, m_3, w_4, m_1)$.
 As man $m_1$ failed to find a new partner,
 we know that the result cannot be stable
 (and indeed by claim~\ref{claimCanGoUpIfNotWosm}, $m_1$ is the best
 stable partner of $w_2$).
 However, we still ``learned something'' along the way:
 $w_4$ accepted a proposal from both $m_1$ and $m_3$,
 and the rejection chain included her twice.
 If we had started the rejection chain from $w_4$,
 we could have actually gotten $MPDA(P_{w_4}(\mu_1)) = (214536)$,
 a new stable matching which differs from $\mu_1$ by the rotation
 $\rho_2 = [(w_4, m_4), (w_5, m_5), (w_3, m_3)]$.
 Thus, we can note this rotation $\rho_2$ and also the fact that
 $w_2$ has reached her best stable partner.

 Because of the pair $(m_2, w_1)$, we know $\rho_1$
 is a type 1 predecessor of $\rho_3$.
 Indeed, $\rho_1$ must be eliminated before $\rho_3$,
 because $w_1$ must be paired to $m_2$ before $\rho_3$ could possible 
 happen\footnote{
   For an example of what would happen if you try to eliminate a type 1
   predecessor before its successor, see figure~\ref{fig:chainVSDepend}.
 }.

 % $w_2$ rejects $m_1$, who then proposes to $w_3$. 
 % $w_3$ rejects $m_3$, who proposes to $w_5$, $w_5$ rejects
 % $m_4$, who proposes to $w_3$. Now $w_3$ has received $2$ proposals better than
 % her match in $\mu_1$, thus by Claim~\ref{claimCoveringCondition}, the result of
 % running $MPDA(P_{w_2}(\mu_1))$ couldn't be a cover of $\mu_1$. However, observe
 % that if we had simply run $MPDA$ that truncated $w_3$'s list, no woman will
 % receive two proposals from men who she strictly prefers to her match in
 % $\mu_1$. Indeed, 
 %  which is a cover to $\mu_1$. We have discovered another rotation

 Finally, the rejection chain of $MPDA(P_{w_1}(\mu_2))$ is simply
 $(w_1, m_2, w_6, m_6, w_1)$.
 This corresponds to the last rotation $\rho_3 = [(w_1, m_2), (w_6, m_6)]$.
 Because of the pair $(m_2, w_3)$,
 we know $\rho_2$ is a type 2 predecessor of $\rho_3$.
 Why does this mean that $\rho_2$ must be eliminated before $\rho_3$?
 Consider trying to eliminate $\rho_3$ before $\rho_2$,
 for example, by running $MPDA(P_{w_1}(\mu_1))$.
 The rejection chain is 
 $(w_1, m_2, w_3, m_3, w_4, m_4, w_5, m_5, w_3, m_2, w_6, m_6, w_1)$,
 and $w_3$ got two proposals and accepted them both.
 The rejection chain between those two proposals corresponds to $\rho_2$.
 Thus, trying to eliminate $\rho_3$ \emph{triggered} the elimination of
 $\rho_2$, even though none of the agents in $\rho_3$ appear in $\rho_2$.

 % First note that, in the run of $MPDA(P_{w_1}(\mu_2))$

 % Note, however, that 
 % $w_1$ rejects $m_2$, who
 % proposes to $w_3$. Since $w_3$ is paired with $m_4$ in $\mu_2$, who she likes
 % better than $m_2$, $w_3$ rejects $m_2$. $m_2$ then moves on to propoes to
 % $w_6$, who rejects $m_6$, who proposes to $w_1$, $w_1$ accepts such a proposal,
 % and $MPDA$ terminates with stable matching $\mu_3 = (265341)$, and we have

 % \textbf{Construct the Precedence Graph} Notice that by
 % Claim~\ref{claimChainsAreToplSorts}. $\rho_1, \rho_2, \rho_3$ are all possible
 % rotations for this example. Moreover, $\rho_1 \rho_2 \rho_3$ is a valid
 % topological order for the rotations.  

 % Now let us identify the predecessor relationship between the rotations using
 % Definition~\ref{def:typedEdges}. $\rho_1$ resulted in the pairing $(m_2, w_1)$,
 % which appears in $\rho_3$, thus $\rho_1$ is a type 1 predecessor of $\rho_3$.
 % On the other hand, $\rho_2$ and $\rho_3$ involves disjoint agents. However, in
 % $\rho_2$, $w_3$ changes its stable partner from $m_3$ (which is ranked below
 % $m_2$) to $m_4$ (which is ranked above $m_2$). Thus $\rho_2$ must happen in
 % order for $m_2$ to be rejected from $w_3$ and propose to $w_6$ in $\rho_3$.
 % Intuitively, we can see this by considering $MPDA(P_{w_1}(\mu_1))$ -- before we
 % could complete $\rho_3$,  we would trigger the elimination of $\rho_2$ (as
 % $m_2$ would propose to $w_3$). We conclude that $\rho_2$ is a type 2
 % predecessor of $\rho_3$. 

\subsection{Algorithm description}

  Given what we know, the high-level interpretation of algorithm~\ref{algMosmToWosm} is fairly intuitive.
  The algorithm starts from the man-optimal stable outcome $\mu_0$.
  Along the way, it maintains a matching $\tilde\mu$, which is always stable,
  and initially set to $\mu_0$.
  From that point on, our goal is to make the smallest possible changes upward
  in the lattice $\L$, i.e. to make it so that each new value of $\tilde\mu$
  covers the old value.

  The algorithm works by picking any woman $\hat w$, and simulating the
  proposals and rejections made in $MPDA(P_{\hat w}(\tilde\mu))$ by having $\hat w$
  ``divorce'' her husband and continuing to run deferred acceptance.
  By claim~\ref{claimCoveringCondition}, we get a new matching which covers
  $\tilde\mu$ (i.e. we find a rotation) if and only if
  no woman receives multiple proposals from a man she prefers to her match in
  $\tilde\mu$. We cannot efficiently guarantee that this will hold for the 
  $\hat w$ that we pick. However, we can get around this 
  issue by using the following trick:
  when a woman considers a new proposal, she decides whether to accept 
  as if she were still matched to her partner in $\tilde\mu$,
  even if she has already accepted a proposal that puts her above that 
  man\footnote{
    Again, figure~\ref{fig:chainVSDepend} provides an example of why this
    is necessary. There, if $w_2$ compared $m_3$ against her match in
    $\mu$ (namely $m_1$) instead of in $\tilde\mu$ (namely $m_2$)
    then we would not find two distinct rotations $\rho_1$ and $\rho_2$.
    Instead, we would find the list $[(w_1,m_1), (w_2,m_2), (w_3,m_3)]$,
    which is not a rotation because eliminating it from $\mu_0$
    does not result in a matching which covers $\mu_0$,
    i.e. we would \emph{miss} some stable matchings in between
    $\mu_0$ and the next matching found in $\tilde\mu$.
  }.
  Then, whenever a woman $w^*$ receives a second
  proposal from a man she prefers to her match in $\tilde\mu$,
  we pause for a minute.
  We consider stable matching corresponding to the execution of $MPDA$
  \emph{between these two proposals which $w^*$ receives}.
  This corresponds to running $MPDA(P_{w^*}(\tilde\mu))$ and getting
  a matching $\mu'$ which covers $\tilde\mu$.
  % to a covering relation from $\tilde\mu$, and thus a rotation.
  So we reset $\tilde\mu$ to the equivalent of $\mu'$ and
  record the corresponding rotation in the graph $G$.

  To prevent the algorithm from doing unnecessary repeated work,
  and to efficiently keep track of when we reach the woman-optimal outcome,
  we maintain a set $S$. Whenever a rejection chain starting with $\hat w$ ends
  in an unmatched woman or unmatched man, we know by 
  claim~\ref{subclaimWomanOptPartner}
  that $\hat w$ cannot receive a better stable match. 
  Because we eliminate all cycles along
  the way, every woman after $\hat w$ on the rejection chain 
  would also trigger this same event.
  Thus, each woman on the current rejection chain has reached their optimal match,
  and can be added to $S$.
  The algorithm runs until $S$ is all of $\W$.

  % REVISE PARAGRAPH THIS NEXT PARAGRAPH:::
  Along the way, we keep track of type 1 and type 2 predecessor using a
  straightforward application of their definition.
  For the type 1 predecessors, it suffices to look at which rotations move the
  men, and create predecessor relations between each successive rotations moving
  the same man.
  % More specifically, we use a variable $pred^1_m$ for each man
  % $m$, which stores the most recently found rotation moving $m$.
  For the type 2 predecessors, intuitively we detect under which conditions
  eliminating $\rho_2$ would \emph{force} the elimination of $\rho_1$, because
  some woman who appears in $\rho_1$ would have accepted a proposal that
  a man $m$ makes as he moves through $\rho_2$.
  % , triggering the divorce chain of $\rho_1$.
  To implement this, we label the men on each woman's preference list,
  putting a label $\rho$ for each woman $w$ in $\rho$
  and each man $m$ such that $\rho$ moves $w$ from below $m$ to above $m$.
  % on the men in each individual
  % woman's preference list whenever $\rho$ moves the woman above those men,
  Then, we accumulate the corresponding rotations as the men make proposals
  (i.e. as a man $m$ gets rejected by some woman $w$,
  we label $m$ with any rotation that moved $w$ from below to above $m$).

  \begin{algorithm}
    \caption{Finding the Rotations and Predecessor Graph}\label{algMosmToWosm}
    \hspace*{\algorithmicindent} \textbf{Input}
      A stable matching instance with men $\M$ and women $\W$ \\
    \hspace*{\algorithmicindent} \textbf{Output}
      A direct graph $G$ on the rotations of the instance
  \begin{algorithmic}[1]
    \State Let $\tilde\mu$ be the man-optimal stable matching from $MPDA$;
      Let $\mu$ be a copy of $\tilde\mu$
    \State For each man $m$, let $R(m)$ be the set of women who rejected $m$
    during the run of $MPDA$.
    \State Let $S$ be the set of unmatched women in $\mu$
    % \State\Comment $S$ will be those women who have upgraded to their optimal stable match
    \State Set $pred^1_m = \emptyset$ for each man $m$ \label{lineInitType1Preds}
      \Comment Store the most recent rotation moving $m$
    \State For each woman $w$ and each man $m$ on $w$'s list, label $m$ 
      in the list with $\emptyset$ \label{lineInitType2Preds}

      \ \Comment Store rotations moving $w$ from \emph{below} to \emph{above} men
    \While { $S \ne \W$ }
      % \State Store $\tilde\mu \leftarrow \mu$
      % \Comment $\tilde\mu$ stores the most recent \emph{stable} match encountered
      \State Pick any $\hat w\in \W\setminus S$
        \label{linePickWomanRejecting}
        \Comment These selections define an execution sequence
      \State Let $m = \mu(\hat w)$; let $V = [ (\hat w, m) ]$
      \State Set $\mu(\hat w) = \emptyset$ and add $\hat w$ to $R(m)$
        \Comment $\hat w$ rejects $m$
      \State Let $pred_{m}^2 = \emptyset$ \label{lineInitAccumedType2PredsWHat}
        \Comment Keep track of predecessor rotations

      % \While { $m \ne \emptyset$ }
      \While { $V \ne [\ ]$ }
        \State Let $w \leftarrow$ \Call{NextAcceptingWoman}{$m$}

        \If { $w=\emptyset$ or $w\in S$ }
        \label{lineFirstDeciciveCase}
          \Comment No stable matching exists rotating partners in $V$
          \State Restore $\mu \leftarrow \tilde\mu$
          \State Add all women in $V$ to $S$; Set $V = [\ ]$
            \label{lineAddFinalizedWomen}
          % \State Set $m = \emptyset$

        \ElsIf {$w$ appears in $V$}
          \Comment New rotation found \label{lineNewRotationFound}
          \If { $w \ne \hat w$ }
            % \State Set $m_{float} \leftarrow \mu(w) $
            \State Swap $\mu(w) \leftrightarrow m$
            \Comment $w$ does not reject $m$ yet
            (see Claim~\ref{claimAlgoRuntime}, item~\ref{itemReproposal})
          \EndIf
          \State \Call{BuildNewRotation}{$w$}
          \label{lineBuildRotCalled}

        \Else % {$w$ is not in $V$ or $w\notin S$} 
          \Comment Continue building rejection chain $V$
          \State Append $(w,\mu(w))$ to the end of $V$
          \State Swap $\mu(w) \leftrightarrow m$; Add $w$ to $R(m)$
            \Comment $w$ rejects $\mu(w)$
          \State Let $pred_{m}^2 = \emptyset$
          \label{lineInitAccumedType2PredsMid}
            \Comment Keep track of predecessor rotations
        \EndIf
        % \State Set $\nu(w) \leftarrow \nu(w)+1$ and $t \leftarrow t+1$

      \EndWhile
    \EndWhile

    \ %
      % \State \Comment If $|\M|<|\W|$, then $m$ has had
      %   his proposal accepted by \emph{some} woman $w$
      % \ ( Now $m >_w \tilde\mu(w)$ (or $w = \emptyset$) )
    % \State // 
    \Function{NextAcceptingWoman}{$m$}
      \State Let $w$ be $m$'s most preferred woman not in $R(m)$ (or $\emptyset$)

      \While { $w\ne \emptyset$ and $\tilde\mu(w) >_w m$ }
        \Comment while $w$ has received a better \emph{stable} match
        \State Add $w$ to $R(m)$
        \Comment $w$ rejects $m$
        \State If $w$ labeled $m$ with $\rho$, add $\rho$ to $pred_m^2$
          \label{lineAccumType2Preds}

        \ \Comment If rotation $\rho$ moved 
          $w$ above $m$, then $\rho$ must precede the current rotation
        \State Update $w$ to $m$'s top woman not in $R(m)$ (or set $w$ to $\emptyset$)
      \EndWhile
      \State Return $w$

    \EndFunction

    \ %
    \Function{BuildNewRotation}{$w$}
      \State Suppose $V = [(w_1,m_1), (w_2,m_2), \ldots,
        (w_k,m_k)]$ with $w = w_\ell$ for some $\ell \le k$
      \State Update $\tilde\mu(w_i) = \mu(w_i)$ for each $i=\ell,\ell+1,\ldots, J$
        \label{lineElimRotationInMatch}
        \Comment Eliminate new rotation $\rho^*$
      \State Remove $\rho^* = [(w_\ell,m_\ell),\ldots,(w_k,m_k)]$ from $V$
        \label{lineRhoStar}
      \State Add rotation $\rho^*$ with type 1 predecessors
        $\bigcup_{i=\ell}^k pred_{m_i}^1$
        \label{lineType1PredsAdded}
        \\ \qquad and type 2 predecessors 
        $\bigcup_{i=\ell}^k pred_{m_i}^2 $ to $G$
        \label{lineType2PredsAdded}
        \label{lineFinalizePredecessors}
      \For {each $i=\ell,\ldots,k$ }
        Set $pred_{m_i}^1 = \rho^*$
        \label{lineNewType1Pred}
      \EndFor
      \For {each $i=\ell,\ldots,k$, and for each man \\
      \qquad\qquad $m$ between 
      $m_i$ and $m_{i-1}$ (or $m_\ell$ and $m_k$ if $i=\ell$) on $w_i$'s list}: \label{lineNewType2MarkStart}
        \State $w_i$ labels $m$ with $\rho^*$
      \EndFor \label{lineNewType2MarkEnd}
      % \For {each man $m$ between 
      %   $m_\ell$ and $m_k$  on $w_k$'s list}:
      %   \State $w_k$ labels $m$ with $\rho^*$
    \EndFunction
  \end{algorithmic}
  \end{algorithm}

%% file: AlgCorrectness.tex
  \subsection{Proof of correctness}

  Our main procedure is given in full detail as algorithm~\ref{algMosmToWosm}.
  An execution sequence of algorithm~\ref{algMosmToWosm} is defined by the
  choice of rejections that the algorithm triggers, more specifically, by each
  choices of the woman $\hat w$ every time we reach line~\ref{linePickWomanRejecting}.
  As in the case of $MPDA$, we will see that the final result is independent of
  these choices and that the total amount of work done is $O(n^2)$.

  % Essentially what will happen is that our analysis will imply that \emph{some}
  % execution of algorithm~\ref{algMosmToWosm} will result in any stable matching
  % you like. Moreover, the different possibilities for execution sequences
  % are essentially ``stored'' in a format which we'll call the \emph{rotation
  % poset}. The result is, after one ``pass'' through a small subset of the stable
  % matching lattice, \emph{every} stable matching will be represented by a
  % compact data structure.

  \begin{claim}\label{claimAccumulateWomenAtTop}
    At any step of algorithm~\ref{algMosmToWosm},
    every woman in the set $S$ has reached her optimal stable match.
  \end{claim}
  \begin{proof}
    We prove this claim by induction on the number of iterations (of the outer
    loop in Algorithm~\ref{algMosmToWosm} from line 7 to line 26) the algorithm
    has run. Let $V_i, \tilde\mu_i, S_i$ denote the value of $V, \tilde\mu, S$
    at the end of iteration $i$ respectively. 
    % Notice that the only time $S$
    % changes is at the end of an iteration in line~\ref{lineAddFinalizedWomen}.
    % Thus in order to prove
    % this claim, it suffices to show that for all $i \geq 0$, every woman in
    % $S_i$ has reached her optimal stable match in
    % $\mu_i$.  

    Firstly, $S_0$ is the set of all unmatched women in MOSM. All women in $S_0$
    has already reached their optimal match by the rural hospital
    theorem~\ref{claimRuralDoctors}.

    Next, assume for $k \leq i$, every woman in $S_k$ has reached her optimal
    stable match at the end of iteration $k$. If $S_{i+1} = S_i$, then the same
    must hold for $S_{i+1}$. If $S_{i+1} \neq S_i$, then at the end of iteration
    $i+1$ the algorithm must have entered the if branch on line 
    \ref{lineFirstDeciciveCase}, which adds
    all women in $V$ to $S$ in line~\ref{lineAddFinalizedWomen}. Let $w(V_i)$ be
    the set of all woman in $V_i$. Then $S_{i+1} = S_{i} \cup w(V_{i+1})$. We
    claim that every woman in $w(V_{i+1})$ must have reached their optimal
    stable matching in $\tilde\mu_{i+1}$. 

    Observe that at any point where a woman $w$ receives two proposals from men
    she prefers to her match in $\tilde\mu$, the algorithm enters the if branch
    on line 20, where subroutine {\sc BuildNewRotation} updates $\tilde\mu$.
    After the update, all women in $V$ have received at most one proposal from
    men she prefers to her match in $\tilde\mu$. Thus only the last woman in the
    rejection chain could have received more than one proposal from men she
    prefers to her in $\tilde\mu_{i+1}$. However, the last woman is either
    $\emptyset$ or in $S_i$, and is never added to $V_{i+1}$. Therefore every
    woman in $w(V_{i+1})$ has only received one proposal from men she prefers to
    her match in $\tilde\mu_{i+1}$. 
    
    Let $V_{i+1} = [(w_1, m_1), (w_2, m_2) ... (w_t, m_t)]$. Since $V$ represent
    the rejection chain in $MPDA(P_{w_1}(\tilde\mu_{i+1}))$, we know that
    $MPDA(P_{w_1}(\tilde\mu_{i+1}))$ results in an unstable matching. Moreover,
    since all woman $w_j \in w(V_{i+1})$ only receives one proposal (from men
    $m_{j-1}$) that she prefers to her match in $\tilde\mu_{i+1}$,
    $MPDA(P_{w_j}(\tilde\mu_{i+1}))$ must have the rejection chain $[(w_j, m_j),
    ... (w_t, m_t)]$, and also result in an unstable matching. By
    Claim~\ref{claimCanGoUpIfNotWosm}, $w_j$ is pair with her optimal stable
    partner in $\tilde\mu_{i+1}$ already.

    % This sequence of rejections is the same as running
    % $MPDA(P_{\hat w}(\tilde\mu))$.
    % By~\ref{claimLatticeExploration}, part~\ref{subclaimWomanOptPartner},
    % $\hat w$ has reached her best stable partner.
    % Moreover, Then for any woman $w$ in $V$, the set of proposals would also
    % lead to $MPDA(P_{w}(\tilde\mu))$ resulting in unstable, thus $w$ has
    % gotten best.
    % Thus, induction.
    % TODOTODO: USE REAL WORDS.

    % If $w=\emptyset$ (i.e. the current man $m$ could not find a new partner)
    % then each $w$ in $V$ has reached her optimal match by CLAIMYO --
    % prove that if truncate and result isn't stable then you're at your optimal.
    % By HERES A CLAIM, for each woman $w\in S$, if $w$ rejected her current
    % partner, then the set of matched agents would change or something -- the
    % last claim needs to be if and only if.
    % But, because each rejection in $V$ triggered the next rejection in $V$,
    % this means that for each woman in $V$, if that woman rejected her current
    % partner then the match would become unstable OR HOWEVER YOU WANT TO
    % PHRASE.
  \end{proof}

  \begin{claim} \label{claimAlgoRuntime}
    Algorithm~\ref{algMosmToWosm} terminates and runs in $O(n^2)$ time.
    % Moreover, the number of rotations found is at most
    % ${n \choose 2}$, and the total number of predecessor relations
    % formed is $O(n^2)$.
  \end{claim}
  \begin{proof}
    % First, we show that there are at most $n(n-1)/2$ rotations.
    % Observe that a pair $(m,w)$ cannot appear in more than one rotation,
    % because men will never propose again to a woman that rejected them,
    % and $w$ must reject $m$ for $(m,w)$ to appear in a rotation.
    % ((HOW TO PROVE THIS AT THIS STAGE IN TIME?? Make sure this is sound
    % with what we know so far. Well kinda obvious -- one moves
    % up the list and the other moves down -- far away from each other)).
    % However, a man $m$ never appears in a rotation with his partner in the
    % woman-optimal stable outcome, so he can appear in a rotation with at most
    % $n-1$ different women. Finally, observe that at least two pairs must be
    % present in each rotation, giving us the desired upper bound.
    Thoughout the algorithm, at each time step, one of the following events happen: (1) a woman rejects a man $m$, (2) a man propose to a woman $w$, either after being rejected or the man repropose after a rotation has been built, (3) A new rotation $\rho$ is extracted, (4) an earlier rotation is added as type $1$ or type $2$ predecessor of new rotation $\rho$ and (5) women are added to $S$. Each event above takes constant time. We show below that the total number of above events is $O(n^2)$.  
    
    \begin{enumerate}[1)]
        \item 
        Woman $w$ can only reject man $m$ once. Thus in total event (1) happen $O(n^2)$ times.
        \item \label{itemReproposal}
        We'd like to say that (as in the case of deferred acceptance) every man
        proposes to every woman at most once. However, there is an important
        exception to this: when a man
        $m$ is tentatively matched to a woman in $\mu$ (i.e. $m$ is in $V$),
        but the woman receives a new proposal which she accepts,
        then the man must propose to the woman again after the corresponding
        rotation has been created. However, this can happen at most once for
        every rotation, and by claim~\ref{claimStablePartnerRotations}, there
        are $O(n^2)$ rotations. So the total number of proposals made is still
        $O(n^2)$.
        \item 
        There are $O(n^2)$ rotations, and each rotation is found at most once. 
        \item
        Throughout the algorithm, a rotation $\rho$ is added to $pred_m^1$ only when $\rho$ changes $m$'s stable partner to some woman $w$. $m$ can only have $O(n)$ different stable partners, and there are only $O(n)$ men. Thus a rotation can only be added as a type $1$ predecessor $O(n^2)$ times. Similarly, we can count the number of type 2 edges. A rotation $\rho$ is added to $pred_m^2$ only when $w$ labeled $m$ with $\rho$, and $w$ rejects $m$. Moreover, each man on $w$'s list receives at most one label. So a rotation is added as a type 2 predecessor $O(n^2)$ times. 
        \item 
        A fixed woman $m$ can only be added to $S$ once. Thus event (5) occur $O(n)$ times.
    % WHAT TO SAY ABOUT TYPE ONE ROTATIONS???
    
    \end{enumerate}

    As we do a constant amount of work for all these events,
    we conclude that the execution time is $O(n^2)$.

    % NOTENOTENOTE: you need to say something about like::
    % At every unit of time (modulo extra work that gets a bound from the paragraphs above) we either pick a new $\hat w$ or add a woman to the existing $V$. Observe that once an woman is added to $V$, they will either receive a better partner or get added to $S$. Thus we'll terminate.
    
  \end{proof}

  We now prove that algorithm~\ref{algMosmToWosm} traverses a maximal chain from
  the man-optimal to the woman-optimal stable outcome.
  Using the theory built up in section~\ref{sectionRotation},
  this will allow us to fairly easily prove that algorithm~\ref{algMosmToWosm}
  correctly outputs all rotations $\Pi$ and the predecessor digraph $G(\Pi)$.

  \begin{claim} \label{claimExecMaxChain}
    During the execution of algorithm~\ref{algMosmToWosm}, let
    the $MPDA(P) = \mu_0, \mu_1, \ldots, \mu_k$ denote the values
    $\tilde\mu$ in order, and let $\rho_0, \rho_1, \ldots, \rho_{k-1}$
    denote the values of $\rho^*$ inserted into $G$ in order.
    Then each $\mu_i$ is a stable matching, $\mu_k$ is the woman-optimal stable
    match, and $\mu_{i+1}$ is the
    elimination of $\rho_i$ from $\mu_i$ for each $i$.
    In other words, $\mu_0, \mu_1, \ldots, \mu_n$ is a maximal chain
    in the stable matching lattice $\L$, with corresponding
    rotation sequence $\rho_0, \rho_1, \ldots, \rho_{k-1}$.

  \end{claim}
  \begin{proof}
    First, we show by induction that each $\mu_i$ is stable.
    Assume that $\tilde\mu$ is stable and the {\sc BuildNewRotation}
    function is called on line~\ref{lineBuildRotCalled}.
    Let $w^*$ denote the current value of $w$ at that line
    (so $w^*$ is the woman in $V$ who just accepted a
    proposal from a man she prefers to her match in $\tilde\mu$),
    and let $\rho^*$ be as on line~\ref{lineRhoStar}.
    Consider the sequence of rejections made after the last time
    $\tilde\mu$ was changed.
    No woman in $\rho^*$ other than (possibly) $w^*$ has received multiple
    proposals from a man she preferred to her match in $\tilde\mu$.
    Furthermore, the sequence of rejections and proposals made
    between the first occurrence
    of $w^*$ in $V$ is exactly those made in
    $MPDA(P_{w^*}(\tilde\mu))$, i.e. the rejection chain of $w^*$ starting from
    $\tilde\mu$ corresponds exactly to $\rho^*$.
    By~\ref{claimCoveringCondition}, $\mu' = MPDA(P_{w^*}(\tilde\mu))$ covers
    $\tilde\mu$. But the new value of $\tilde\mu$, set on
    line~\ref{lineElimRotationInMatch}, is exactly $\mu'$
    (which is exactly the elimination of rotation $\rho^*$ from $\tilde\mu$).
    Thus, $\mu_{i+1}$ covers $\mu_i$ for each $i$,
    and $\mu_{i+1}$ is the elimination of $\rho_i$ from $\mu_{i}$.

    Algorithm~\ref{algMosmToWosm} terminates only when $S$ is all of $\W$.
    But by claim~\ref{claimAccumulateWomenAtTop},
    $S$ consist only of women who have reached their optimal stable match in
    $\tilde\mu$. Thus, when the algorithm terminates, $\tilde\mu$ is the
    woman-optimal stable match.

    % Alt addition to this theorem:

    % Moreover, for every maximal chain $\mu_0,\ldots,\mu_k$
    % in $\L$, there exists an execution
    % sequence of algorithm~\ref{algMosmToWosm} which results in $\tilde\mu$
    % taking values $\mu_0,\ldots,\mu_k$ in order.

    % Finally, let $\mu_0 < \mu_1 < \ldots < \mu_k$ be any maximal chain in $\L$.
    % Then $\mu_i$ covers $\mu_{i-1}$ for each $i$.
    % We start with $\tilde\mu = \mu_0$ (the man optimal stable match),
    % so inductively assume that we reach line~\ref{linePickWomanRejecting}
    % while $\tilde\mu = \mu_i$.
    % Pick $\hat w$ to be any woman who receives a better match in $\mu_i$
    % than in $\mu_{i-1}$.
    % By claim~\ref{claimCanGoUpIfNotWosm},
    % the result of $MPDA(P_{\hat w}(\mu_i))$ will be exactly $\mu_{i+1}$.
    % By claim~\ref{claimCoveringCondition}, no woman during the execution of
    % $MPDA(P_{\hat w}(\mu_i))$ will receive multiple proposals from a man she
    % prefers to her match in $\mu_i$.
    % Thus, the sequence of proposals and rejections occurring in
    % algorithm~\ref{algMosmToWosm} will be exactly the same as in
    % $MPDA(P_{\hat w}(\mu_i))$, and the next value of $\tilde\mu$ will be
    % set to $\mu_{i+1}$.
  \end{proof}

  % MAYBE we should prove that agents visit all their stable partners in order,
  % but probably we shouldn't do so here -- it'll be a corollary later. So I've
  % COMMENTED stuff out.

  \begin{claim} \label{claimAlgoCorrect}
    Every rotation of $\Pi$ is found and put into $G$
    over the course of algorithm~\ref{algMosmToWosm}.
    Furthermore, if $\rho_1$ is a predecessor of $\rho_2$ (type 1 or type 2)
    then $\rho_1$ will be found before $\rho_2$.
    Moreover, the set of predecessors in $G$ of every rotation $\rho$ 
    % is the same across all executions of
    % algorithm~\ref{algMosmToWosm}, and these
    are exactly the type 1 and type 2
    predecessors defined above, i.e. $G = G(\Pi)$.
  \end{claim}
  \begin{proof}
    The first two statements now readily follow from the previous claim
    and claim~\ref{claimChainsAreToplSorts} (parts
    \ref{subclaimAllRotationsAppear} and
    \ref{subclaimPredecesorsBefore} respectively).

    We know that for each rotation $\rho^*$,
    each predecessor of $\rho^*$ has certainly been found
    by the time we construct $\rho^*$.
    We now show that the predecessors of $\rho^*$ are appropriately marked.
    The type 1 predecessors are added on line~\ref{lineType1PredsAdded},
    and they are exactly the most recently found rotations moving man $m$
    (as consistently updated on line~\ref{lineNewType1Pred}).
    But the most recent rotation moving $m$ must be the unique rotation
    which moved $m$ to his current match $w$, where $(w,m)$ appears in $\rho^*$.
    So the type 1 predecessors of $\rho^*$ are accurately marked.

    For the type 2 predecessors, consider a man $m$ who appears in $\rho^*$
    and moves from $w$ to $w'$. From the time $m$ entered $V$
    (in line~\ref{lineInitAccumedType2PredsWHat}
    or~\ref{lineInitAccumedType2PredsMid}), we added to $pred^2_m$
    (on line~\ref{lineAccumType2Preds}) each
    rotation $\rho$ such that $m$ was rejected by a woman $w$ and
    $\rho$ moved $w$ from below $m$ to above $m$
    (all such rotations are predecessors of $\rho^*$,
    and thus have already been found and appropriately marked on
    lines~\ref{lineNewType2MarkStart} to~\ref{lineNewType2MarkEnd}).
    As $m$ is rejected by each woman between $w$ and $w'$ on his list,
    this covers all possible type 2 predecessors of $\rho^*$,
    so the type 2 predecessors are accurately marked
    on line~\ref{lineType2PredsAdded}.

    % . As $m$ proposes to exactly 

    % For the type 1 predecessors, there is at most one for each man who appears
    % in $\rho^*$, and because 
    % The last statement follows from the fact that the definition of
    % type 1 and type 2 predecessor is exactly what we
    % used to assign $pred_m^1$ and $pred_m^2$ sets in the algorithm
  \end{proof}

  We can now immediately conclude from claims~\ref{claimAlgoRuntime} 
  and \ref{claimAlgoCorrect}
  that the rotation predecessor graph $G(\Pi)$ can be
  correctly and efficiently computed.
  \begin{theorem} \label{theoremAlgRuntime}
    Algorithm \ref{algMosmToWosm} computes $G(\Pi)$ in $O(n^2)$ time. 
  \end{theorem}

%% file: GusfieldContrast.tex
The primary difference between our approach and the original is that
\cite{GusfieldStableStructureAlgs89} starts by
defining a partial order on the set of rotations $\Pi$,
then constructs a directed acyclic graph $G$ and proves that the 
transitive closure of $G$ gives the correct order on $\Pi$.
% (going so far as to
% build up a theory of ``the minimal differences of a rings of sets'').
In contrast, we start with $G = G(\Pi)$ and are able to prove directly that
$G(\Pi)$ represents the set of all stable matchings.
% Then, when they turn to algorithmic considerations, they define a graph
% equivalent to our $G(\Pi)$.
% By starting with $G(\Pi)$, we achieve a more streamlined and
% intuitive presentation.
% Most of our proofs are essentially equivalent to those from
% \cite{GusfieldStableStructureAlgs89}, but the proof approach for
% claim~\ref{claimToplSortsAreChains} is new (it replaces lemma 3.2.4 in
% \cite{GusfieldStableStructureAlgs89}, and its usage of
% claim~\ref{claimToplSortsInterchange} is one of the crucial steps
% allowing us to skip a lot of the ``theory building''
% and work with $G(\Pi)$ without first defining a partial order on $\Pi$).
% We also focus more explicitly on the concept of
% \emph{covering relations} in the stable matching lattice
% (though the lattice-theoretic properties
% needed are quite simple).

% 2) We base our proofs on GI89's graph G(M) instead of on a partial order defined
% directly on Pi. This is the primary "simplifying step".  I give more detail
% below.

Our definition of rotations is similar to that
of~\cite{GusfieldStableStructureAlgs89}.
The only difference
is that we define rotations to be the difference between two stable matchings
where one covers the other, whereas~\cite{GusfieldStableStructureAlgs89}
essentially defines rotations as rejection chains of MPDA
(where no woman receives multiple proposals from a
better man than her old match) without explicitly mentioning MPDA
(GI89 must prove that rotations give covering relations in their
lemma 2.5.5, although they do not explicitly use the term ``covering'').

\cite{GusfieldStableStructureAlgs89} defines a partial order relation on $\Pi$
via the partial order
relation on their ``minimal difference on a ring of sets''.  In turn, the ordering
on the minimal differences is defined via the lattice ordering, restricted to
the ``irreducible elements (other than the man-optimal)''.
Building the theory of these partial orders occupies
\cite{GusfieldStableStructureAlgs89} for the entirety of their Chapter 2
(pages 67 to 102).
We believe that intuition is lost through these layers of definition.

There are two steps to proving that the transitive closure of $G$ is exactly
the partial order on $\Pi$.  First, every edge in $G$ should be related in $\Pi$.
We find that the essence of the proof in~\cite{GusfieldStableStructureAlgs89}
(lemma 3.2.3) goes through without referencing a partial order on $\Pi$ at all.
We capture this with our claim~\ref{claimChainsAreToplSorts},
which shows that any valid sequence of rotations respects $G$.

The second step is to show that every relation in $\Pi$ is in the transitive
closure of $G$.
\cite{GusfieldStableStructureAlgs89} heavily relies on the existing 
partial order on $\Pi$ for this step of their proof (lemma 3.2.4)
because they prove that any \emph{immediate}
predecessor (i.e. a covering relation) in $\Pi$ must be related in $G$.
% TODO: footnote some stuff about how they do this.

% (This
% follows by looking at what must happen if mu exposes rho (and not rho') and
% mu/rho exposes rho'. You can then prove that rho is a type 1 or type 2
% predecessor of rho'.)

We are able to skip this crucial reliance on an order on $\Pi$ via our
claim~\ref{claimToplSortsAreChains},
which proves that every ordering of rotations which respects $G$ corresponds
to a maximal chain in the stable matching lattice.
The key lemma we use is claim~\ref{claimToplSortsInterchange},
which shows that if two rotations are adjacent in some topological sort
and do not have an edge between them, then that pair of rotations
can be swapped (while maintaining the property that the topological
sort corresponds to a maximal chain).

In the end, our method amounts to showing that the collection of topological
sorts of $G$ is in a bijection with the linear extensions of $\Pi$,
and using a ``swapping'' lemma like claim~\ref{claimToplSortsInterchange}
to prove the correspondence in one direction.
The linear extensions of $\Pi$ are particularly natural objects in our case
(as they correspond to maximal chains in the stable matching lattice).
To the best of our knowledge, this strategy for proving that a graph gives
a certain transitive closure has not been used before.
It may be useful in other situations involving distributive lattices where
maximal chains in the lattice are easy to reason about.

As in~\cite{GusfieldStableStructureAlgs89}, we deliberately avoid mentioning
Birkhoff's representation theorem. While this classical theorem immediately
shows the existence of \emph{some} partial order which represents any
distributive lattice, it does not show how to find this representation
or give any structure regarding what the elements of the partial order are.
% algorithmic applications need rotations!!
Indeed, Birkhoff's theorem by itself could not even show that the partial order
$\Pi$ is polynomial-size.

% The only new observation we need is that, if two "adjacent" rotations are NOT
% type 1 or type 2 predecessors, they they can be eliminated in either order.  It
% follows that any sequence of rotations that respects G(M) is a valid sequence of
% rotation eliminations in the stable matching lattice.  Thus we essentially prove
% that G(M) characterizes the partial order on Pi, without every having to define
% the partial order on Pi in the first place.

% ====
% They key step that allows us to directly use G(M) is our claims 4.6 and 4.7,
% which works by show that any order that respects G(M) can be a sequence of rotation
% eliminations.

%% file: GusfieldLabelsCounterex.tex
The \emph{minimal-differences}
algorithm of \cite{GusfieldStableStructureAlgs89} 
(more precisely, figure 3.2 on page 110) correctly identifies all
of the rotations in a stable matching instance. However, there is a slight error
in the construction of the order relations for the rotation poset.
In particular, once the rotations are found (via an algorithm essentially
equivalent to our algorithm~\ref{algMosmToWosm},
but without keeping track of predecessor relations), they propose Algorithm~\ref{algRotationOrderGI} as shown below.

\begin{algorithm} 
  \caption{Construct predecessor relations}\label{algRotationOrderGI}
\begin{algorithmic}[1]
  \For {Each rotation $\rho$ and pair $(m_i,w_i)\in \rho$}
    \State Label $w_i$ in $m_i$'s preference list with a type 1 $\rho$ label
    \For {Each $m$ strictly between $m_i$ and $m_{i-1}$ on $w_i$'s list}
      % \Comment $w_i$ upgrades $m_i$ to $m_{i-1}$
      \State Label $w_i$ in $m$'s preference list with a type 2 $\rho$ label
    \EndFor
  \EndFor
  \For {Each man $m$}
    \State Set $\rho^* = \emptyset$
    \For {Each woman $w$ on $m$'s preference list, in order}
      \If {$w$ has a type 1 label of $\rho$}
        \If { $\rho^* \ne \emptyset$ } 
          Add $\rho^*$ as a predecessor of $\rho$
        \EndIf
        \State Set $\rho^* = \rho$
      \EndIf
      \If { $w$ has a type 2 label of $\rho$}
        \If {$\rho^*\ne \emptyset$}
          Add $\rho$ as a predecessor of $\rho^*$
        \EndIf
      \EndIf
    \EndFor
  \EndFor
\end{algorithmic}
\end{algorithm}

The idea behind this algorithm is reasonable:
certainly the type 1 labels in any man's chain should be related in the poset
(as the man needs to reach a certain partner before the next rotation can be
found).
Furthermore, suppose a type 2 label $\rho_2$ is between two type 1 labels,
$\rho_1$ and $\rho_3$. We know that $\rho_1$ moved $m$ from
his partner in $\rho_1$ to his partner is $\rho_3$,
as men propose in their preference order at most once to each woman.
Along the rejection chain from his partner in $\rho_1$
to his partner in $\rho_3$, $m$ would propose to some woman $w$ in $\rho_2$, and
$w$ likes $m$ better than her match in $\rho_2$. Thus, the rejection chain of
$\rho_1$ will certainly trigger $\rho_2$, and $\rho_1$ must be a predecessor of
$\rho_2$.

However, the above reasoning fails in certain cases. Namely, in the case
where there is a type 1 label $\rho^*$ followed by a type 2 label $\rho$
on woman $w$,
but in $\rho^*$ man $m$ does not move from above $w$ to below $w$.
In this case, the rejection sequence $\rho^*$ does not actually trigger rotation
$\rho$.

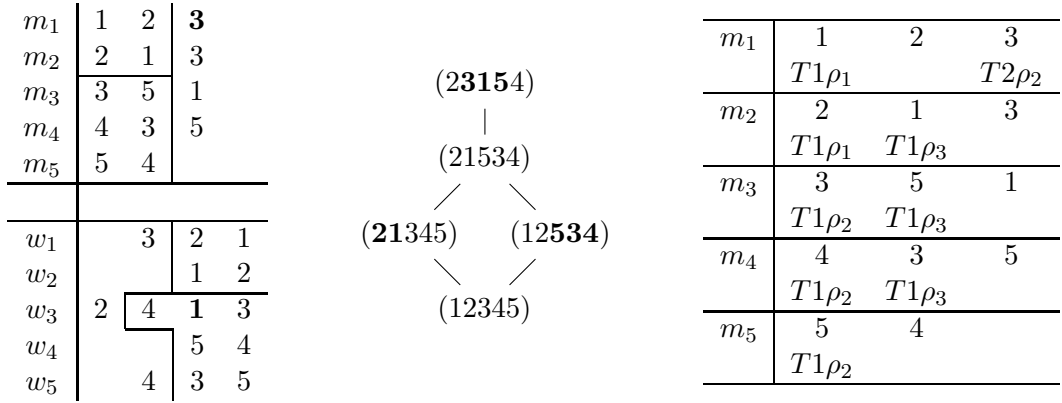
\begin{figure}[htp] 
  \centering
  \begin{tabular}{c | c c c c}
    $m_1$ & 1 &\multicolumn{1}{c|}{2} & \textbf{3} \\
    $m_2$ & 2 &\multicolumn{1}{c|}{1} & 3 \\
    \cline{2-3}
    $m_3$ & 3 &\multicolumn{1}{c|}{5} & 1 \\
    $m_4$ & 4 &\multicolumn{1}{c|}{3} & 5 \\
    $m_5$ & 5 &\multicolumn{1}{c|}{4} & \\
    \hline
    & & & & \\
    \hline
    $w_1$ & & \multicolumn{1}{c|}{3} & 2 & 1 \\
    $w_2$ & & \multicolumn{1}{c|}{ } & 1 & 2 \\
    \cline{3-5}
    $w_3$ & \multicolumn{1}{c|}{2} & 4 & \textbf{1} & 3 \\
    \cline{3-3}
    $w_4$ & & & \multicolumn{1}{|c}{5} & 4 \\
    $w_5$ & & 4 & \multicolumn{1}{|c}{3} & 5 \\
  \end{tabular}
  \qquad
  \begin{tabular}{c}
    \begin{tikzpicture}[scale=1]
      \node (zero) at (0,-1) {$(12345)$};
      \node (l) at (-1,0) {$({\bf21}345)$};
      \node (r) at (1,0) {$(12{\bf534})$};
      \node (b) at (0,1) {$(21534)$};
      \node (one) at (0,2) {$(2{\bf315}4)$};
      \draw (zero) -- (l) -- (b) -- (one);
      \draw (zero) -- (r) -- (b);
    \end{tikzpicture} 
  \end{tabular}
  \qquad
  \begin{tabular}{c| ccc}
    \hline
    $m_1$ & 1 & 2 & 3 \\
      & $T1\rho_1$ & & $T2\rho_2$ \\
    \hline
    $m_2$ & 2 & 1 & 3 \\
      & $T1\rho_1$ & $T1\rho_3$ & \\
    \hline
    $m_3$ & 3 & 5 & 1 \\
      & $T1\rho_2$ & $T1\rho_3$ & \\
    \hline
    $m_4$ & 4 & 3 & 5 \\
      & $T1\rho_2$ & $T1\rho_3$ & \\
    \hline
    $m_5$ & 5 & 4 & \\
      & $T1\rho_2$ & & \\
    \hline
  \end{tabular}
  \caption{A tricky case for algorithm~\ref{algRotationOrderGI},
    including the labeled preference lists. }
  \label{figTrickyPreds}
\end{figure}

For a concrete counterexample, consider the stable matching instance
in figure~\ref{figTrickyPreds}, drawn alongside its lattice $\L$ (with matchings
written as the by writing the partner of $w_1, w_2,\ldots, w_5$ in order).
The rotations of this instance are 
$\rho_1=[(1,2),(2,1)]$, $\rho_2=[(3,3),(4,4),(5,5)]$, and
$\rho_3=[(2,1),(4,3),(3,5)]$, and $\rho_1$ and $\rho_2$ are both type 1
predecessors of $\rho_3$.
The labeled preference list of the men, given by applying
algorithm~\ref{algRotationOrderGI}, is also drawn in figure~\ref{figTrickyPreds}.

% \begin{wrapfigure}{r}{5cm}
%   \caption{The labeled men's preference lists}
%   \label{figTrickyLists}
% \end{wrapfigure}

The above algorithm causes $\rho_1$ to be marked as a predecessor of $\rho_2$, even
though they are independent and both exposed in the man-optimal stable matching.
Our algorithm~\ref{algMosmToWosm} circumvents this problem
by storing the type 1 and
type 2 labels in different places, and detecting the required orderings
on the rotations more directly. 

Another way around this problem, which is more similar to
\cite{GusfieldStableStructureAlgs89}'s algorithm~\ref{algRotationOrderGI},
would be to write ``type 1 end markers'' for the final type 1 label
in each man's preference list (note that this problem can only happen
for type 2 labels after the final type 1 label, because if there is another type
1 label after the type 2 label, the man must actually move below the woman 
$w$ where the type 2 label was marked).
More specifically, for the last type 1 label on
$m$'s list, say of a rotation $\rho$, mark ``type 1 end'' on the woman $w$ for
which $\rho$ moves $m$ to $w$. Then, ignore any type 2 labels after the ``type 1
end'' mark.

% More specifically, the type 1 labels are stored,
% one for each man, and the type 2 labels appear on the preference lists of the
% women.